\documentclass[runningheads]{llncs}

\newif\ifdraft\draftfalse

\newif\ifproceedings\proceedingsfalse

\usepackage{amsthm} 

\ifdraft
\usepackage[draft]{commenting}
\else
\usepackage[nompar]{commenting}
\fi

\declareauthor{lo}{lo}{purple}
\authorcommand{lo}{comment}

\declareauthor{hp}{hp}{blue}
\authorcommand{hp}{comment}

\declareauthor{cm}{cm}{green!50!black}
\authorcommand{cm}{comment}
\setdefaultauthor{cm}

\declareauthor{dw}{dw}{orange}
\authorcommand{dw}{comment}

\usepackage[normalem]{ulem}

\DeclareMathAlphabet{\mathpzc}{OT1}{pzc}{m}{it} 

\usepackage[bookmarks,unicode,colorlinks=true]{hyperref}

\usepackage[bbgreekl]{mathbbol}

\DeclareSymbolFontAlphabet{\mathbbl}{bbold}
\DeclareMathSymbol\bbDelta  \mathord{bbold}{"01}

\usepackage{wrapfig}

\usepackage{url}

\usepackage{etex}
\usepackage{amssymb}
\usepackage{stmaryrd}
\usepackage{mathtools}
\usepackage{dsfont}
\usepackage[inline]{enumitem}

\usepackage{relsize} 
\usepackage{bussproofs} 
\usepackage{proof}
\usepackage{algorithm} 
\usepackage{listings} 
\usepackage{tabularx,tabulary,multirow} 
\usepackage{tikz}
\usepackage{etoolbox} 
\usepackage{graphicx}
\usepackage{subcaption}
\usepackage[misc]{ifsym}

\usepackage{xcolor}
\newcommand{\mathhighlight}[1]{\colorbox{black!13}{$\displaystyle #1$}}

\definecolor{codegreen}{rgb}{0,0.6,0}
\definecolor{codegray}{rgb}{0.4,0.4,0.4}
\definecolor{codepurple}{rgb}{0.58,0,0.82}

\usetikzlibrary{cd,trees,shapes,patterns}

\newtheorem{assumption}{Assumption}
\usepackage{cleveref}

\Crefname{theorem}{Thm.}{Theorems}
\Crefname{corollary}{Cor.}{Corollary}
\Crefname{proposition}{Prop.}{Propositions}
\Crefname{claim}{Claim}{Claims}
\Crefname{definition}{Def.}{Definitions}
\Crefname{fact}{Fact}{Facts}
\Crefname{conj}{Conjecture}{Conjectures}
\Crefname{example}{Ex.}{Examples}
\Crefname{remark}{Remark}{Remarks}
\Crefname{convention}{Convention}{Conventions}
\Crefname{lemma}{Lem.}{Lemmas}
\Crefname{assumption}{Assumption}{Assumptions}
\Crefname{section}{Sec.}{Sections}
\Crefname{appendix}{Appendix}{Appendices}
\Crefname{figure}{Fig.}{Figures}
\crefname{theorem}{Thm.}{Theorems}
\crefname{corollary}{Cor.}{Corollary}
\crefname{proposition}{Prop.}{Propositions}
\crefname{claim}{Claim}{Claims}
\crefname{definition}{Def.}{Definitions}
\crefname{fact}{Fact}{Facts}
\crefname{conj}{Conjecture}{Conjectures}
\crefname{example}{Ex.}{Examples}
\crefname{remark}{Remark}{Remarks}
\crefname{convention}{Convention}{Conventions}
\crefname{lemma}{Lem.}{Lemmas}
\crefname{assumption}{Assumption}{Assumptions}
\crefname{section}{Sec.}{Sections}
\crefname{appendix}{Appendix}{Appendices}
\crefname{figure}{Fig.}{Figures}
\crefname{equation}{Eq.}{Equations}

\usepackage{thm-restate}

\usepackage{filecontents}
\usepackage{pgfplots, pgfplotstable}
\usepgfplotslibrary{statistics}
\pgfplotsset{compat=1.16}

\newcommand{\Nat}{\mathbb{N}}
\newcommand{\Real}{\mathbb{R}}
\newcommand{\Rational}{\mathbb{Q}}

\newcommand{\Borel}{\mathcal{B}}
\newcommand{\Function}[5]{
  \begin{align*}
    #1: \quad #2 & \;\longrightarrow\; #3 \\
              #4 & \;\longmapsto\;     #5
  \end{align*}
}
\newcommand{\Case}[3]{
  \begin{cases}
    #1 & \text{if } #2,\\
    #3 & \text{otherwise.}
  \end{cases}
}
\newcommand{\inv}[1]{{#1}^{-1}}
\newcommand{\charfn}[1]{\mathbbl{1}_{#1}}
\newcommand{\interior}[1]{\mathring{#1}}
\newcommand{\closure}[1]{\overline{#1}}

\newcommand{\dom}[1]{\mathsf{dom}(#1)}
\newcommand{\domnp}{\mathsf{dom}}
\newcommand{\salgebra}{\Sigma}

\newcommand{\PCFReal}{\mathsf{R}}
\newcommand{\tyarrow}{\Rightarrow}
\newcommand{\PCF}[1]{\underline{#1}}
\newcommand{\If}[3]{\mathsf{if}\big(#1, #2, #3\big)}
\newcommand{\Sample}{\mathsf{sample}}
\newcommand{\Score}[1]{\mathsf{score}(#1)}
\newcommand{\Y}[1]{\mathsf{Y}{#1}}
\newcommand{\Fail}{\mathsf{fail}}

\newcommand{\terms}{\Lambda}
\newcommand{\closedterms}{\Lambda^0}
\newcommand{\values}{\Lambda_v}
\newcommand{\closedvalues}{\Lambda^0_v}

\newcommand{\primitives}{\mathcal{F}}

\newcommand{\walk}{\mathsf{walk}}
\newcommand{\passZero}{\walk}
\newcommand{\letin}[1]{\mathsf{let}\ {#1}\ \mathsf{in}\ }
\newcommand{\Normal}[2]{\mathcal{N}(#1,#2)}

\newcommand{\dist}{s}
\newcommand{\sco}{z}

\newcommand{\terma}{M}
\newcommand{\termb}{N}
\newcommand{\termc}{L}
\newcommand{\typea}{\sigma}
\newcommand{\typeb}{\tau}

\newcommand{\traces}{\mathbb{S}}
\newcommand{\trace}{\boldsymbol{s}}
\newcommand{\traceb}{\boldsymbol{s'}}
\newcommand{\traceseq}[1]{[#1]}
\newcommand{\trmeasure}{\mu_{\traces}}
\newcommand{\emptytrace}{{[]}}
\newcommand{\evalcon}{E}
\newcommand{\valuea}{V}
\newcommand{\redexa}{R}

\makeatletter

\RequirePackage[bookmarks,unicode,colorlinks=true]{hyperref}%
   \def\@citecolor{blue}%
   \def\@urlcolor{blue}%
   \def\@linkcolor{blue}%

\def\orcidID#1{\smash{\href{http://orcid.org/#1}{\protect\raisebox{-1.25pt}{\protect\includegraphics{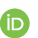}}}}}

\newcommand{\dotDelta}{{\vphantom{\Delta}\mathpalette\d@tD@lta\relax}}
\newcommand{\d@tD@lta}[2]{%
  \ooalign{\hidewidth$\m@th#1\mkern-1mu\cdot$\hidewidth\cr$\m@th#1\Delta$\cr}%
}

\DeclareFontFamily{OMX}{MnSymbolE}{}
\DeclareSymbolFont{MnLargeSymbols}{OMX}{MnSymbolE}{m}{n}
\SetSymbolFont{MnLargeSymbols}{bold}{OMX}{MnSymbolE}{b}{n}
\DeclareFontShape{OMX}{MnSymbolE}{m}{n}{
    <-6>  MnSymbolE5
   <6-7>  MnSymbolE6
   <7-8>  MnSymbolE7
   <8-9>  MnSymbolE8
   <9-10> MnSymbolE9
  <10-12> MnSymbolE10
  <12->   MnSymbolE12
}{}
\DeclareFontShape{OMX}{MnSymbolE}{b}{n}{
    <-6>  MnSymbolE-Bold5
   <6-7>  MnSymbolE-Bold6
   <7-8>  MnSymbolE-Bold7
   <8-9>  MnSymbolE-Bold8
   <9-10> MnSymbolE-Bold9
  <10-12> MnSymbolE-Bold10
  <12->   MnSymbolE-Bold12
}{}

\let\llangle\@undefined
\let\rrangle\@undefined
\DeclareMathDelimiter{\llangle}{\mathopen}%
                     {MnLargeSymbols}{'164}{MnLargeSymbols}{'164}
\DeclareMathDelimiter{\rrangle}{\mathclose}%
                     {MnLargeSymbols}{'171}{MnLargeSymbols}{'171}

\newcommand{\p@thmlisti}{\perh@ps{\thetheorem}}
\protected\def\perh@ps#1#2{\textup{#1.#2}}

\makeatother
\newcommand{\contra}{R'}
\newcommand{\config}[3]{\left\langle{#1,\allowbreak#2,\allowbreak#3}\right\rangle}
\newcommand{\skterms}{\mathsf{SK}}

\newcommand{\red}{\to}
\newcommand{\redstar}{\to^*}
\newcommand{\leb}{\mathrm{Leb}}
\newcommand{\weightfn}{\mathsf{weight}}
\newcommand{\valuefn}{\mathsf{value}}
\newcommand{\oper}[1]{\llbracket #1 \rrbracket}

\newcommand{\defn}[1]{\textbf{\em #1}}
\newcommand{\diff}{\mathrm{d}}
\newcommand{\pdf}{\mathrm{pdf}}

\newcommand{\reals}{\mathbb{R}}

\newcommand{\boundary}{\partial}
\newcommand{\sred}{\Rightarrow}

\newcommand{\sconfig}[3]{\llangle #1,#2,#3 \rrangle} 

\newcommand{\defeq}{\coloneqq}
\newcommand{\termi}{\mathbb T}
\newcommand{\termii}{\mathbb T^{\mathsf{int}}}
\newcommand{\pReal}{\Real_{\geq 0}}

\newcommand{\from}{:}
\newcommand{\pto}{\rightharpoonup}
\newcommand{\scon}{\mathpzc E}
\newcommand{\sval}{\mathpzc V}

\newcommand{\sredex}{\mathpzc R}
\newcommand{\scontra}{\mathpzc R'}

\newcommand{\sterm}{\mathpzc M}
\newcommand{\stermb}{\mathpzc N}
\newcommand{\stermc}{\mathpzc L}
\newcommand{\Ifleq}[3]{\mathsf{if}\big(#1\leq 0, #2, #3\big)}
\newcommand{\const}{\mathpzc 1}
\newcommand{\sweight}{\mathpzc w}
\newcommand{\pop}{\mathcal {F}}

\newcommand{\tand}{\text{ and }}

\newcommand{\tup}[3]{\boldsymbol {#1}_{#2:#3}}
\newcommand{\tupf}[1]{\boldsymbol {#1}}
\newcommand{\seva}[1]{\left\|#1\right\|}

\newcommand{\reptn}[2]{\vartriangleright_{(\tupf #1,\tupf #2)}}

\newcommand{\tow}{\text{otherwise }}
\newcommand\concat{\ensuremath{\mathbin{+\mkern-10mu+}}}

\newcommand{\termss}[2]{\terms_{(#1, #2)}}
\newcommand{\delay}[1]{\setlength{\fboxsep}{.5\fboxsep}\boxed{#1}}
\newcommand{\delaynew}[1]{\,\underline{\underline{\,#1\,}}\,}

\newcommand{\conc}[1]{\left\lfloor #1\right\rfloor}
\newcommand{\trdiv}{\traces_{\mathsf{div}}}
\newcommand{\defeqs}{\equiv}

\newcommand{\args}{(\tupf r,\tupf s)}
\newcommand{\argsb}{(\tupf r,\tupf {s'})}
\newcommand{\argsc}{(\tupf r,\trace\concat\traceb)}

\newcommand{\tr}{\mathbb T}
\newcommand{\trb}{U}
\newcommand{\trmax}{\tr_{\mathsf{max}}}
\newcommand{\trterm}{\tr_{\mathsf{term}}}
\newcommand{\trtermM}{\tr_{M,\mathsf{term}}}

\newcommand{\trtermi}{\tr_{\mathsf{term}}^{\mathsf{int}}}
\newcommand{\trext}{\tr_{\mathsf{stuck}}}
\newcommand{\trexti}{\tr_{\mathsf{stuck}}^{\mathsf{int}}}
\newcommand{\trprefi}{\tr_{\mathsf{pref}}^{\mathsf{int}}}
\newcommand{\trpref}{\tr_{\mathsf{pref}}}
\newcommand{\trmaxr}{\traces_{\tupf r,\mathsf{max}}}
\newcommand{\trtermr}{\traces_{\tupf r,\mathsf{term}}}

\newcommand{\calM}{\mathcal{M}}

\begin{document}

%
%
%
%
%
%

\title{Densities of Almost Surely Terminating Probabilistic Programs are Differentiable Almost Everywhere}

\titlerunning{Densities of A.S.~Terminating Programs are Differentiable A.E.} 


\author{Carol Mak \orcidID{0000-0002-6512-2864} \and
C.-H.~Luke Ong \orcidID{0000-0001-7509-680X} \and
Hugo Paquet \orcidID{0000-0002-8192-0321} \and
Dominik Wagner$^{\text{(\Letter)}}$ \orcidID{0000-0002-2807-8462}}

\authorrunning{C.~Mak, C.-H.~L.~Ong, H.~Paquet and D.~Wagner}


\institute{Department of Computer Science, University of Oxford, Oxford, UK
\email{\{pui.mak,luke.ong,hugo.paquet,dominik.wagner\}@cs.ox.ac.uk}}

\maketitle

\begin{abstract}




We study the differential properties of higher-order statistical probabilistic programs with recursion and conditioning.
Our starting point is an open problem posed by Hongseok Yang: what class of statistical probabilistic programs have densities that are differentiable almost everywhere?
To formalise the problem, we consider Statistical PCF (SPCF), an extension of call-by-value PCF with real numbers, and constructs for sampling and conditioning.
We give SPCF a sampling-style operational semantics \`a la Borgstr\"om et al.,
and study the associated weight (commonly referred to as the density) function and value function on the set of possible execution traces.

Our main result is that almost surely terminating SPCF programs, generated from a set of primitive functions (e.g.~the set of analytic functions) satisfying mild closure properties, have weight and value functions that are almost everywhere differentiable.
We use a stochastic form of symbolic execution to reason about almost everywhere differentiability.
A by-product of this work is that almost surely terminating \emph{deterministic} (S)PCF programs with real parameters denote functions that are almost everywhere differentiable.

Our result is of practical interest, as almost everywhere differentiability of the density function is required to hold for the correctness of major gradient-based inference algorithms.


\end{abstract}

\section{Introduction}
\label{sec:intro}

\emph{Probabilistic programming} refers to a set of tools and techniques for the systematic use of programming languages in Bayesian statistical modelling.
Users of probabilistic programming --- those wishing to make \changed[lo]{inferences or predictions} --- 
\textbf{(i)} encode their domain knowledge in program form;
\textbf{(ii)} \emph{condition} certain program variables based on observed data; and
\textbf{(iii)} make a query.
The resulting code is then passed to an \emph{inference engine} which performs the necessary computation to answer the query, usually following a generic approximate Bayesian inference algorithm. (In some recent systems \cite{pyro,DBLP:conf/pldi/Cusumano-Towner19}, users may also write their own inference code.) The Programming Language community has contributed to the field by developing formal methods for probabilistic programming languages (PPLs), seen as usual languages enriched with primitives for \textbf{(i)} sampling and \textbf{(ii)} conditioning.
(The query \textbf{(iii)} can usually be encoded as the return value of the program.)

It is crucial to have access to reasoning principles in this context.
The combination of these new primitives with the traditional constructs of programming languages leads to a variety of new computational phenomena,
and a major concern is the \emph{correctness of inference}:
given a query, will the algorithm converge, in some appropriate sense, to a correct answer?
In a \emph{universal} PPL (i.e.~one whose underlying language is Turing-complete),
this is not obvious: the inference engine must account for a wide class of programs,
\changed[dw]{going beyond the more well-behaved models found in many of the current statistical applications}.
Thus the design of inference algorithms, and the associated correctness proofs, are quite delicate.
It is well-known, for instance, that in its original version the popular lightweight Metropolis-Hastings algorithm \cite{DBLP:journals/jmlr/WingateSG11} contained a bug affecting the result of inference \cite{hur2015provably,Kiselyov2016a}.

Fortunately, research in this area benefits from decades of work on the semantics of programs with random features,
starting with pioneering work by Kozen \cite{DBLP:conf/focs/Kozen79} and Saheb-Djahromi \cite{saheb1978probabilistic}.
Both operational and denotational models have recently been applied to the validation of inference algorithms:
see e.g.~\cite{hur2015provably,DBLP:conf/icfp/BorgstromLGS16} for the former and \cite{scibior2017denotational,castellan2019probabilistic} for the latter.
There are other approaches, e.g.~using refined type systems \cite{lew2019trace}.

Inference algorithms in probabilistic programming are often based on the concept of \emph{program trace}, because the operational behaviour of a program is parametrised by the sequence of random numbers it draws along the way. \changed[hp]{Accordingly a probabilistic program has an associated \emph{value function} which maps traces to output values.
But the inference procedure relies on another function on traces, commonly called the \emph{density}\footnote{For some readers this terminology may be ambiguous; see \cref{remarkdensities} for clarification.} of the program, which records a cumulative likelihood for the samples in a given trace.
Approximating a normalised version of the density is the main challenge that inference algorithms aim to tackle.
We will formalise these notions: in \cref{sec:semantics} we demonstrate how the value function and density of a program are defined in terms of its operational semantics.
}

\subsubsection{Contributions.}

The main result of this paper is that both the density and value function are \emph{differentiable almost everywhere} (that is, everywhere but on a set of measure zero),
provided the program is \emph{almost surely terminating} in a suitable sense.
Our result holds for a universal language with recursion and higher-order functions.
We emphasise that it follows immediately that \emph{purely deterministic programs with real parameters} denote functions that are almost everywhere differentiable.
This class of programs is important, because they can express
machine learning models which rely on gradient descent \cite{DBLP:journals/corr/abs-2006-06903}.

\dw{I think it's better to swap the paragraphs like this:}
This result is of practical interest, because many modern inference algorithms are ``gradient-based'':
they exploit the derivative of the density function in order to optimise the approximation process.
This includes the well-known methods of Hamiltonian Monte-Carlo \cite{DuaneKPR87,neal2011mcmc} and stochastic variational inference \cite{DBLP:journals/jmlr/HoffmanBWP13,RGB14,blei2017variational,KRGB15}.
But these techniques can only be applied when the derivative exists ``often enough'', and thus, in the context of probabilistic programming, almost everywhere differentiability is often cited as a requirement for correctness \cite{DBLP:conf/aistats/ZhouGKRYW19,DBLP:journals/pacmpl/LeeYRY20}.
The question of which probabilistic programs satisfy
this property was selected by Hongseok Yang in his FSCD 2019 invited lecture {\cite{DBLP:conf/rta/Yang19} as one of three open problems in the field of semantics for probabilistic programs.

\changed[hp]{

Points of non-differentiability exist largely because of \emph{branching}, which typically arises in a program when the control flow reaches a conditional statement. Hence our work is a study of the connections between the traces of a probabilistic program and its branching structure.
To achieve this we introduce \emph{stochastic symbolic execution}, a form of operational semantics for probabilistic programs, designed to identify sets of traces corresponding to the same control-flow branch.
Roughly, a reduction sequence in this semantics corresponds to a control flow branch, and the rules additionally provide for every branch a symbolic expression of the trace density, parametrised by the outcome of the random draws that the branch contains. 
}
\changed[dw]{We obtain our main result in conjunction with a careful analysis of the branching structure of almost surely terminating programs.}

\subsubsection{Outline.}

We devote \cref{sec:probprog} to a more detailed introduction to the problem of trace-based inference in probabilistic programming, and the issue of differentiability in this context.
In \cref{sec:semantics}, we present a trace-based operational semantics to Statistical PCF, a prototypical higher-order functional language previously studied in the literature.
This is followed by a discussion of differentiability and almost sure termination of programs (\cref{sec:diff}).
In \cref{sec:symbolic} we define the ``symbolic'' operational semantics required for the proof of our main result, which we present in \cref{subsec:main result}.
We discuss related work and further directions in \cref{sec:conclusion}.

\ifproceedings
For the extended version of the paper refer to \cite{MOPW20}.
\fi

\section{Probabilistic Programming and Trace-Based Inference}
\label{sec:probprog}

In this section we give a short introduction to probabilistic programs and the densities they denote, and we motivate the need for gradient-based inference methods. Our account relies on classical notions from measure theory, so we start with a short recap.

\subsection{Measures and Densities}
\label{subsec:measures and densities}
A \defn{measurable space} is a pair $(X, \Sigma_X)$ consisting of a set together with a \defn{$\sigma$-algebra} of subsets, i.e.~$\Sigma_X \subseteq \mathcal{P}(X)$ contains $\emptyset$ and is closed under complements and countable unions and intersections. Elements of $\Sigma_X$ are called \defn{measurable \changed[cm]{sets}}.
A \defn{measure} on $(X, \Sigma_X)$ is a function $\mu : \Sigma_X \to [0, \infty]$ satisfying $\mu(\emptyset) = 0$, and $\mu(\bigcup_{i \in I} U_i) = \sum_{i \in I} \mu(U_i)$
for every countable family $\{ U_i \}_{i\in I}$ of pairwise disjoint measurable subsets. A (possibly partial) function $X \pto Y$ is \defn{measurable} if for every $U \in \Sigma_Y$ we have $f^{-1}(U) \in \Sigma_X$.

The space $\reals$ of real numbers is an important example.
The \changed[cm]{(Borel)} $\sigma$-algebra $\Sigma_\reals$ is the smallest one containing all intervals $[a, b)$, and the \defn{Lebesgue measure} $\leb$ is the unique measure on $(\reals, \Sigma_\reals)$ satisfying $\leb([a, b)) = b -a$.
For measurable spaces $(X, \Sigma_X)$ and $(Y, \Sigma_Y)$, the \defn{product $\sigma$-algebra} $\Sigma_{X \times Y}$ is the smallest one containing all $U \times V$, where $U \in \Sigma_X$ and $V \in \Sigma_Y$. So in particular we get for each $n \in \Nat$ a space $(\reals^n, \Sigma_{\reals^n})$, and additionally there is a unique measure $\leb_n$ on $\reals^n$ satisfying $\leb_n(\prod_{i} U_i) = \prod_i{\leb(U_i)}$.

When a function $f : X \to \reals$ is measurable and non-negative and $\mu$ is a measure on $X$, for each $U \in \Sigma_X$ we can define the \defn{integral} $\int_U (\diff\mu) f \in [0, \infty]$. Common families of probability distributions on the reals (Uniform, Normal, etc.) are examples of measures on $(\reals, \Sigma_\reals)$.
Most often these are defined in terms of \emph{probability density functions} with respect to the Lebesgue measure, meaning that for each $\mu_D$ there is a measurable function $\pdf_D : \reals \to \pReal$ which determines it: $\mu_D(U) = \int_U (\diff \, \leb)\ \pdf_D.$ As we will see,  density functions such as $\pdf_D$ have a central place in Bayesian inference.

Formally, if $\mu$ is a measure on a measurable space $X$, a \defn{density} for $\mu$ with respect to another measure $\nu$ on $X$ (most often $\nu$ is the Lebesgue measure) is a measurable function $f : X \to \reals$ such that
$\mu(U) = \int_U (\diff \nu) f$ for every $U \in \Sigma_X$. In the context of the present work, an \emph{inference algorithm} can be understood as a method for approximating a distribution of which we only know the density up to a normalising constant. In other words, if the algorithm is fed a (measurable) function $g : X \to \reals$, it should produce samples approximating the  probability measure $U \mapsto \frac{\int_U (\diff \nu) g}{\int_X (\diff \nu) g}$ on $X$.

We will make use of some basic notions from topology: given a topological space $X$ and an set $A \subseteq X$, the \defn{interior} of $A$ is the largest open set $\interior A$ contained in $A$. Dually the \defn{closure} of $A$ is the smallest closed set $\closure A$ containing $A$, and the \defn{boundary} of $A$ is defined as $\boundary A\defeq \closure A\setminus \interior A$. Note that for all $U\subseteq\Real^n$, all of $\interior U$, $\closure U$ and $\boundary U$ are measurable (in $\Sigma_{\Real^n}$).

\subsection{Probabilistic Programming: a (Running) Example}

Our running example is based on a random walk in $\pReal$.

The story is as follows: a pedestrian has gone on a walk on 
a certain semi-infinite street
(i.e.~extending infinitely on one side), where she may periodically change directions.
Upon reaching the end of the street she has forgotten her starting point,
only remembering that she started no more than 3km away.
Thanks to an odometer, she knows the total distance she has walked is 1.1km, although there is a small margin of error.
Her starting point can be inferred using probabilistic programming, via the program in \cref{fig:example_code}.

\begin{figure}[ht]
\begin{subfigure}[b]{0.6\textwidth}
\begin{lstlisting}
(*returns total distance travelled*)
let rec walk start =
  if (start <= 0) then
    0
  else
    (*each leg < 1km*)
    let step = Uniform(0, 1) in
    if (flip ()) then
      (*go towards +infty*)
      step + walk (start+step)
    else
      (*go towards 0*)
      step + walk (start-step)
in
(*prior*)
let start = Uniform(0, 3) in
let distance = walk start in
(*likelihood*)
score ((pdfN distance 0.1) 1.1);
(*query*)
start
\end{lstlisting}
\subcaption{\label{fig:example_code} Running example in pseudo-code.}
\end{subfigure}
\begin{subfigure}[b]{0.35\textwidth}
  \centering
   \includegraphics{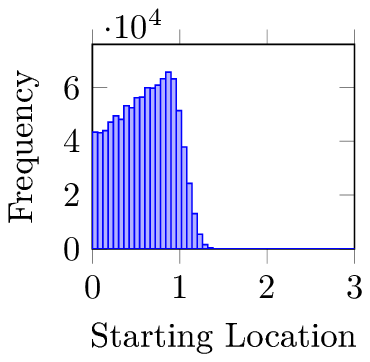}
    \subcaption{\label{fig:example_hist} Resulting histogram.}

\end{subfigure}
\caption{\label{fig:example} Inferring the starting point of a random walk on $\pReal$, in a PPL.}
\end{figure}

The function \lstinline{walk} in \cref{fig:example_code} is a recursive simulation of the random walk: note that in this model a new direction is sampled after at most 1km.
Once the pedestrian has travelled past 0 the function returns the total distance travelled.
The rest of the program first specifies a \emph{prior distribution} for the starting point,
representing the
pedestrian's
belief --- uniform distribution on $[0, 3]$ --- before observing the distance measured by the odometer.
After drawing a value for \lstinline{start} the program simulates a random walk,
and the execution is weighted (via \lstinline{score}) according to how close \lstinline{distance} is to the observed value of $1.1$.
The return value is our query: it indicates that we are interested in the \emph{posterior} distribution on the starting point.

The histogram in \cref{fig:example_hist} is obtained by sampling repeatedly from the posterior of a \changed[lo]{Python} model of our running example.
It shows the mode of the pedestrian's starting point to be around the \changed[lo]{0.8km} mark.

To approximate the posterior, inference engines for probabilistic programs often proceed indirectly and operate on the space of \emph{program traces}, rather than on the space of possible return values.
By \emph{trace}, we mean the sequence of samples drawn in the course of a particular run, one for each random primitive encountered.
Because each random primitive (qua probability distribution) in the language comes with a density, given a particular trace we can compute a coefficient as the appropriate product.
We can then multiply this coefficient by all \lstinline{score}s encountered in the execution, and this yields a (\emph{weight}) function, mapping traces to the non-negative 
reals, over which the chosen inference algorithm may operate.
This indirect approach is more practical, and enough to answer the query, since every trace unambiguously induces a return value.

\begin{remark}
  \label{remarkdensities}
  In much of the probabilistic programming literature (e.g.~\cite{DBLP:journals/pacmpl/LeeYRY20,DBLP:conf/aistats/ZhouGKRYW19,DBLP:conf/rta/Yang19}, including this paper), the above-mentioned weight function on traces is called the \emph{density} of the probabilistic program.
  This may be confusing: as we have seen, a probabilistic program induces a posterior probability distribution on return values, and it is natural to ask whether this distribution admits a probability density function (Radon-Nikodym derivative) w.r.t.~some base measure.
  This problem is of current interest \cite{DBLP:conf/popl/BhatAVG12,DBLP:journals/lmcs/BhatBGR17,DBLP:conf/icfp/IsmailS16} but unrelated to the present work.
\end{remark}

\subsection{Gradient-Based Approximate Inference}

\changed[lo]{Some of the most influential and practically important inference algorithms make use of the gradient of the density functions they operate on, when these are differentiable.
Generally the use of gradient-based techniques allow for much greater efficiency in inference.} \lo{* References needed here. *}


\changed[lo]{A popular example is the Markov Chain Monte Carlo algorithm known as Hamiltonian Monte Carlo (HMC) \cite{DuaneKPR87,neal2011mcmc}.
Given a density function $g : X \to \reals$,
HMC samples are obtained as the states of a Markov chain by (approximately)}
\changed[cm]{simulating Hamilton's equations}
\changed[lo]{via an integrator that uses the gradient $\nabla_x \, g(x)$.
Another important example is (stochastic) variational inference \cite{DBLP:journals/jmlr/HoffmanBWP13,RGB14,blei2017variational,KRGB15}, which transforms the posterior inference problem to an optimisation problem.
This method takes two inputs: the posterior density function of interest $g : X \to \reals$, and a function $h : \Theta \times X \to \reals$; typically, the latter function is a member of an expressive and mathematically well-behaved family of densities that are parameterised in $\Theta$.
The idea is to use stochastic gradient descent to find the parameter $\theta \in \Theta$ that minimises the ``distance'' (typically the Kullback–Leibler divergence) between $h(\theta, -)$ and $g$, relying on a suitable estimate of the gradient of the objective function.
When $g$ is the density of a probabilistic program (the \emph{model}), $h$ can be specified as the density of a second program (the \emph{guide}) whose traces have additional $\theta$-parameters.
The gradient of the objective function is then estimated in one approach (score function \cite{DBLP:conf/aistats/RanganathGB14}) by computing the gradient $\nabla_\theta \, h(\theta, x)$, and in another (reparameterised gradient \cite{DBLP:journals/corr/KingmaW13,DBLP:conf/icml/RezendeMW14,DBLP:conf/icml/TitsiasL14}) by computing the gradient $\nabla_x \, g(x)$.}

In probabilistic programming, the above inference methods must be adapted to deal with the fact that in a universal PPL, the set of random primitives encountered can vary between executions, and traces can have arbitrary and unbounded dimension; \changed[lo]{moreover, the density function of a probabilistic program is generally not (everywhere) differentiable.
Crucially} these adapted algorithms are only valid when the input densities are \emph{almost everywhere} differentiable \cite{DBLP:conf/aistats/ZhouGKRYW19,Nishimura2020a,DBLP:conf/nips/0001YY18}; this is the subject of this paper.

\medskip

Our main result (\cref{thm:mainres}) states that the weight function and value function of almost surely terminating SPCF programs are almost everywhere differentiable.
This applies to our running example: the program in \cref{fig:example_code} (expressible in SPCF using primitive functions that satisfy \cref{ass:pop} -- see \cref{example: SPCF syntax}) is almost surely terminating.

\section{Sampling Semantics for Statistical PCF}
\label{sec:semantics}


In this section, we present a simply-typed 
statistical probabilistic programming language with recursion
and its operational semantics.

\subsection{Statistical PCF}
\label{subsec:stat PCF}

\begin{figure}[t]
  \begin{align*}
    \typea,\typeb & ::=
    \PCFReal \mid \typea \tyarrow \typeb \\
    \terma,\termb,\termc & ::=
    y \mid
    \PCF{r} \mid
    \PCF{f}(\terma_1,\dots,\terma_\ell) \mid
    \lambda y\ldotp\terma \mid
    \terma\,\termb \mid
    \Y{\terma} \mid
    \If{\termc \leq 0}{\terma}{\termb} \\
    &\mid
    \mathhighlight{\Sample} \mid
    \mathhighlight{\Score{\terma}}
  \end{align*}
  $$
    \AxiomC{\vphantom{$4$}}
    \UnaryInfC{$\Gamma \vdash \Sample:\PCFReal$}
    \DisplayProof
    \qquad
    \AxiomC{$\Gamma \vdash \terma:\PCFReal$}
    \UnaryInfC{$\Gamma \vdash \Score{\terma}:\PCFReal$}
    \DisplayProof
    \qquad
    \AxiomC{$
      \Gamma \vdash \terma : (\typea \tyarrow \typeb) \tyarrow (\typea \tyarrow \typeb)
    $}
    \UnaryInfC{$
      \Gamma \vdash \Y{\terma} : \typea \tyarrow \typeb
    $}
    \DisplayProof
  $$
  \caption{
    Syntax of SPCF, where $r \in \Real$, $x,y$ are variables, and
   $f:\Real^n \pto
 \Real$ ranges over a set $\pop$ of partial, measurable
   \defn{primitive functions}
   (see \cref{par:primitive}).
  }
  \label{fig:SPCF syntax}
\end{figure}

\defn{Statistical PCF} (SPCF) is higher-order probabilistic programming with recursion in purified form.
The terms and part of the (standard) typing system of SPCF are presented in \cref{fig:SPCF syntax} \footnote{In \cref{fig:SPCF syntax} and in other figures, we highlight the elements that are new or otherwise noteworthy.}.
In the rest of the paper we write $\tupf{x}$ to represent a sequence of variables $x_1,\dots,x_n$,
$\terms$ for the set of SPCF terms, and
$\closedterms$ for the set of closed SPCF terms.
In the interest of readability, we sometimes use pseudo code (e.g.~\cref{fig:example_code}) in the style of Core ML to express SPCF terms.

SPCF is a statistical probabilistic version of call-by-value
PCF \cite{DBLP:journals/tcs/Scott93,DBLP:conf/fsttcs/Sieber90} with reals as the ground type.
The probabilistic constructs of SPCF are relatively standard (see for example \cite{DBLP:conf/esop/Staton17}): the sampling construct $\Sample$ draws from $\mathcal{U}(0, 1)$, the standard uniform distribution with end points $0$ and $1$;
the scoring construct $\Score{\terma}$ enables conditioning on observed data by multiplying the weight of the current execution with the (non-negative) real number denoted by $\terma$.
\changed[lo]{Sampling from other real-valued distributions can be obtained from $\mathcal{U}(0, 1)$ by applying the inverse of the distribution's cumulative distribution function.}

Our SPCF is an (inconsequential) variant of CBV SPCF \cite{DBLP:journals/pacmpl/VakarKS19} and a (CBV) extension of PPCF \cite{DBLP:journals/pacmpl/EhrhardPT18} with scoring;
it may
be viewed as a simply-typed version of the untyped probabilistic languages
of \cite{DBLP:conf/icfp/BorgstromLGS16,DBLP:conf/esop/CulpepperC17,DBLP:journals/pacmpl/WandCGC18}.

\begin{example}[Running Example ${\sf Ped}$]
  \label{example: SPCF syntax}
  We express in SPCF the example in \cref{fig:example_code}.
  \begin{align*}
    {\sf Ped} &\equiv
    \begin{pmatrix*}[l]
      \letin{x = \Sample \cdot \PCF{3}} \\
      \letin{d = \walk\,x} \\
      \letin{w = \Score{\PCF{\pdf_{\Normal{1.1}{0.1}}}(d)}} x
    \end{pmatrix*} \qquad\text{where} \\
    \walk & \equiv
    \Y{}
    \begin{pmatrix*}[l]
      \lambda f x.& \mathsf{if}\ x \leq \PCF{0}\ \mathsf{then}\ \PCF{0} \\
                  & \mathsf{else}\
                    \begin{pmatrix*}[l]
                      \letin{s = \Sample} \\
                      \If{(\Sample \leq \PCF{0.5})}
                      {\big(s + f(x+s)\big)}{\big(s + f(x-s)\big)}
                    \end{pmatrix*}
    \end{pmatrix*}
  \end{align*}
  The let construct, $\letin{x = \termb}{\terma}$, is syntactic sugar for the term $(\lambda x.\terma)\,\termb$; and $\pdf_{\Normal{1.1}{0.1}}$, the density function of the normal distribution with mean $1.1$ and variance $0.1$, is a primitive function.
  To enhance readability we use infix notation and omit the underline for standard functions such as addition.
\end{example}

\subsection{Operational Semantics}
\label{subsec:operational semantics}

The execution of a probabilistic program generates a \emph{trace}: a sequence containing the values sampled during a run.
Our operational semantics captures this dynamic perspective.
This is closely related 
to the treatment in \cite{DBLP:conf/icfp/BorgstromLGS16} which, following \cite{DBLP:conf/focs/Kozen79}, views a probabilistic program as a deterministic program parametrized by the sequence of random draws made during the evaluation.

\subsubsection{Traces.}


Recall that in our language, $\Sample$ produces a random value in the open unit interval; accordingly a \defn{trace}
is a finite sequence of elements of $(0, 1)$.
We define a \defn{measure space $\traces$ of traces} to be the set $\bigcup_{n\in\Nat} (0, 1)^n$, equipped with the standard disjoint union $\sigma$-algebra, and the sum of the respective (higher-dimensional) Lebesgue measures. Formally, writing $\traces_n \defeq (0,1)^n$, we define:
$$
  \traces \defeq\left(
    \displaystyle \bigcup_{n\in\Nat} \traces_n,
    \left\{\displaystyle \bigcup_{n\in\Nat} U_n \mid U_n \in \salgebra_{\traces_n}\right\},
    \trmeasure
  \right)
  \ \text{and} \
  \trmeasure\left(\displaystyle \bigcup_{n\in\Nat} U_n\right) \defeq\sum_{n\in\Nat} \leb_n(U_n).
$$
Henceforth we write traces as lists,
such as $[0.5, 0.999, 0.12]$;
the empty trace as $\emptytrace$; and
the concatenation of traces $\trace, \trace' \in \traces$ as $\trace \concat \trace'$.

More generally, to account for open terms, we define, for each $m \in \Nat $, the measure space
\[
\Real^m \times \traces \defeq\left(
    \displaystyle \bigcup_{n\in\Nat} \Real^m \times \traces_n,
    \left\{\displaystyle \bigcup_{n\in\Nat} V_n \mid V_n \in \salgebra_{\Real^m \times \traces_n}\right\},
    \mu_{\Real^m \times \traces}
  \right)
\]
where $\mu_{\Real^m \times \traces}\Big(\bigcup_{n\in\Nat} V_n\Big) \defeq \sum_{n\in\Nat} \leb_{m+n}(V_n).$
To avoid clutter, we will elide the subscript from $\mu_{\Real^m \times \traces}$ whenever it is clear from the context.



\subsubsection{Small-Step Reduction.}
Next, we define the
\defn{values} (typically denoted $\valuea$),
\defn{redexes} (typically $\redexa$) and
\defn{evaluation contexts} (typically $\evalcon$):
\begin{align*}
  \valuea & ::=
  \PCF{r} \mid \lambda y.M \\
  \redexa & ::=
            (\lambda y\ldotp\terma)\,\valuea \mid
            \PCF{f}(\PCF{r_1},\dots,\PCF{r_\ell}) \mid
            \Y{(\lambda y\ldotp\terma)} \mid
            \Ifleq{\PCF{r}}{\terma}{\termb} \mid
            \Sample \mid
            \Score{\PCF{r}}\\
  \evalcon & ::=
             [] \mid
             \evalcon\,\terma \mid
             (\lambda y.\terma)\,\evalcon \mid
             \PCF{f}(\PCF{r_1}, \dots, \PCF{r_{i-1}},\evalcon,\terma_{i+1},\dots,\terma_\ell) \mid
             \Y{\evalcon} \\
          & \mid
            \Ifleq{\evalcon}{\terma}{\termb} \mid
            \Score{\evalcon}
\end{align*}
We write $\values$ for the set of SPCF values, and $\closedvalues$ for the set of closed SPCF values.



It is easy to see that every closed SPCF term $\terma$ is
either a value,
or there exists a unique pair of context $\evalcon$ and redex $\redexa$ such that
$\terma \equiv \evalcon[\redexa]$.

We now present the operational semantics of SPCF as a rewrite system of \defn{configurations}, which are triples of the form $\config{\terma}{w}{\trace}$ where $M$ is a closed SPCF term, $w \in \pReal$ is a \defn{weight}, and $\trace \in \traces$ a trace.
(We will sometimes refer to these as the \emph{concrete} configurations, in contrast with the \emph{abstract} configurations of our symbolic operational semantics, see \cref{sec:sred}.)

\begin{figure}

  \noindent\defn{Redex Contractions}:
  \begin{align*}
    \config{(\lambda y\ldotp\terma)\,\valuea}{w}{\trace} & \red
                                                      \config{\terma[\valuea/y]}{w}{\trace} \\
    \config{\PCF{f}(\PCF{r_1},\dots,\PCF{r_\ell})}{w}{\trace}
                                                    & \red
                                                      \config{\PCF{f(r_1,\dots,r_\ell)}}{w}{\trace}
                                                      \tag{if $(r_1,\dots,r_\ell) \in \dom{f}$} \\
    \config{\PCF{f}(\PCF{r_1},\dots,\PCF{r_\ell})}{w}{\trace}& \red\Fail\tag{if $(r_1,\dots,r_\ell)\not\in \dom{f}$} \\
    \config{\Y{(\lambda y.\terma)}}{w}{\trace} & \red
                                                 \config{\lambda z.\terma[\Y{(\lambda y.\terma)}/y]\,z}{w}{\trace}
                                                 \tag{for fresh variable $z$} \\
    \config{\Ifleq{\PCF{r}}{\terma}{\termb}}{w}{\trace} & \red
                                                          \config{\terma}{w}{\trace} \tag{if $r \leq 0$} \\
    \config{\Ifleq{\PCF{r}}{\terma}{\termb}}{w}{\trace} & \red
                                                          \config{\termb}{w}{\trace} \tag{if $r > 0$} \\
    \config{\Sample}{w}{\trace} & \red
                                  \config{\PCF{r}}{w}{\trace\concat [r]} \tag{for some $r \in (0,1)$} \\
    \config{\Score{\PCF{r}}}{w}{\trace} & \red
                                          \config{\PCF{r}}{r\cdot w}{\trace} \tag{if $r \geq 0$}\\
    \config{\Score{\PCF{r}}}{w}{\trace}&\red\Fail\tag{if $r < 0$}
  \end{align*}
  \noindent\defn{Evaluation Contexts}:
  \[
    \infer{\config{E[\redexa]}{w}{\trace} \red \config{E[\contra]}{w'}{\trace'}}{\config{\redexa}{w}{\trace} \red \config{\contra}{w'}{\trace'}}
    \qquad\qquad
    \infer{\config{E[\redexa]}{w}{\trace} \red\Fail}{\config{\redexa}{w}{\trace} \red\Fail}
  \]

  \caption{Operational small-step semantics of SPCF}
  \label{fig:opsem}
\end{figure}

The small-step reduction relation $\red$ is defined in \cref{fig:opsem}.
In the rule for $\Sample$, a random value $r \in (0,1)$ is generated and recorded in the trace, while the weight remains unchanged: in a uniform distribution on $(0, 1)$ each value is drawn with likelihood 1. In the rule for $\Score{\PCF{r}}$, the current weight is multiplied by non-negative $r\in\reals$: typically this reflects the likelihood of the current execution given some observed data.
Similarly to \cite{DBLP:conf/icfp/BorgstromLGS16} we reduce terms which cannot be reduced in a reasonable way (i.e.\ scoring with negative constants or evaluating functions outside their domain) to $\Fail$.

\begin{example}
  \label{example: reduction} We present a possible reduction sequence for the program in \cref{example: SPCF syntax}:
  \begin{align*}
  \config{{\sf Ped}}{1}{\emptytrace}  & \redstar
    \config{
      \begin{pmatrix*}[l]
        \letin{x = \PCF{0.2} \cdot \PCF{3}} \\
        \letin{d = \walk\,x} \\
        \letin{w = \Score{\PCF{\pdf_{\Normal{1.1}{0.1}}}(d)}} x
      \end{pmatrix*}
    }{1}{\traceseq{0.2}} \\
    & \redstar
    \config{
      \begin{pmatrix*}[l]
        \letin{d = \walk\,\PCF{0.6}} \\
        \letin{w = \Score{\PCF{\pdf_{\Normal{1.1}{0.1}}}(d)}} \PCF{0.6}
      \end{pmatrix*}
    }{1}{\traceseq{0.2}} \\
    & \redstar
    \config{
      \letin{w = \Score{\PCF{\pdf_{\Normal{1.1}{0.1}}}(\PCF{0.9})}} \PCF{0.6}
    }{1}{\traceseq{0.2,0.9,0.7}}
    \tag{$\star$}
    \\
    & \redstar
    \config{
      \letin{w = \Score{\PCF{0.54}}} \PCF{0.6}
    }{1}{\traceseq{0.2,0.9,0.7}}\\
    & \redstar
    \config{\PCF{0.6}}{0.54}{\traceseq{0.2,0.9,0.7}}
  \end{align*}
  In this execution, the initial $\Sample$ yields $\PCF{0.2}$, which is appended to the trace.
  At step $(\star)$, we assume given a reduction sequence $\config{\walk\,\PCF{0.6}}{1}{\traceseq{0.2}} \redstar$ \\$\config{\PCF{0.9}}{1}{\traceseq{0.2,0.9,0.7}}$;
  this means that in the call to $\walk$, ${0.9}$ was sampled as the the step size and
  ${0.7}$ as the direction factor; this makes the new location ${-0.3}$, which is negative, so the return value is ${0.9}$.
In the final step, we perform \emph{conditioning} using the likelihood of observing ${0.9}$
  given the data ${1.1}$: the $\Score$ expression updates the current weight using the
  the density of $0.9$ in the normal distribution with parameters $(1.1, 0.1)$.
\end{example}

\subsubsection{Value and Weight Functions.}
Using the relation $\red$, we now aim to reason more globally about probabilistic programs in terms of the traces they produce. Let $\terma$ be an SPCF term with free variables amongst $x_1,\dots,x_m$ of type $\PCFReal$. Its \defn{value function} $\valuefn_{\terma} : \Real^m \times \traces \to \closedvalues \cup \{\bot\}$ returns, given values for each free variable and a trace, the output value of the program, if the program terminates in a value. The \defn{weight function} $\weightfn_{\terma} : \Real^m \times \traces \to \pReal$ returns the final weight of the corresponding execution. Formally:
\begin{align*}
\valuefn_{\terma} (\tupf{r}, \trace) &\defeq
\begin{cases}
V & \hbox{if $\config{\terma[\tupf{\PCF{r}}/\tupf{x}]}{1}{\emptytrace} \red^*
    \config{\valuea}{w}{\trace}$}\\
\bot & \tow
\end{cases}
\end{align*}%
\begin{align*}
\weightfn_{\terma} (\tupf{r}, \trace) &\defeq
\begin{cases}
w & \hbox{if $\config{\terma[\tupf{\PCF{r}}/\tupf{x}]}{1}{\emptytrace} \red^*
    \config{\valuea}{w}{\trace}$}\\
0 & \tow
\end{cases}
\end{align*}
For closed SPCF terms $\terma$ we just write $\weightfn_\terma(\tupf s)$ for $\weightfn_\terma(\emptytrace,\tupf s)$ (similarly for $\valuefn_\terma$), and
it follows already from \cite[Lemma 9]{DBLP:conf/icfp/BorgstromLGS16} that the functions $\valuefn_\terma$ and $\weightfn_\terma$ are measurable (see also \cref{sec:bdiff}).

Finally, every closed SPCF term $\terma$ has an associated \defn{value measure}
$${{\oper{\terma}}:{\salgebra_{\closedvalues}}\to{\pReal}}$$ defined by
$\oper{\terma}{(U)} \defeq
{
  \int_{\inv{\valuefn_\terma}(U)}
  \diff \trmeasure \;
  \weightfn_\terma
}$. This corresponds to the denotational semantics of SPCF in the $\omega$-quasi-Borel space model via computational adequacy \cite{DBLP:journals/pacmpl/VakarKS19}.



\medskip

\changed[lo]{Returning to \cref{remarkdensities}, what are the connections, if any, between the two types of density of a program? To distinguish them, let's refer to the weight function of the program, $\weightfn_M$, as its \emph{trace density}, and the Radon-Nikodyn derivative of the program's value-measure, $\frac{d \oper{\terma}}{d \nu}$ where $\nu$ is the reference measure of the measurable space $\salgebra_{\closedvalues}$, as the \emph{output density}. \lo{Or \emph{return-value density}?}
Observe that, for any measurable function $f : {\closedvalues} \to [0, \infty]$,
$\int_{\closedvalues} d \, \oper{\terma} \, f
=
\int_{\valuefn_\terma^{-1}(\closedvalues)} d \trmeasure \; \weightfn_\terma \cdot (f \circ \valuefn_{\terma})
=
\int_{\traces} d \trmeasure \; \weightfn_\terma \cdot (f \circ \valuefn_{\terma})
$
(because if $\trace \not\in \valuefn_\terma^{-1}(\salgebra_{\closedvalues})$ then $\weightfn_\terma(\trace) = 0$).
It follows that we can express any expectation w.r.t.~the output density $\frac{d \oper{\terma}}{d \nu}$ as an expectation w.r.t.~the trace density $\weightfn_\terma$.
If our aim is, instead, to generate samples from $\frac{d \oper{\terma}}{d \nu}$ then we can simply generate samples from $\weightfn_\terma$, and deterministically convert each sample to the space $(\closedvalues, \salgebra_{\closedvalues})$ via the value function $\valuefn_M$.
In other words, if our intended output is just a sequence of samples, then our inference engine does not need to concern itself with the consequences of change of variables.
}

\lo{TODO:

- PROVE: For any closed term $\terma$, $\oper{\terma}$ is an s-finite measure (s-finite kernel if open), and not $\sigma$-finite in general.

- QUESTION: What assumptions on the primitive functions / constraints on $\Score{\cdot}$ guarantee finiteness of measure $\oper{\terma}$ (which is equivalent to integrability of $\weightfn_M$), for a.s.~terminating $M$?

- Reference measure for $\salgebra_{\closedvalues}$.

- Trace density, as opposed to ouput density, is important in practice.

- QUESTION: When does $\oper{\terma}$ have a R-N derivative? By contrast, trace density is always well-defined.

- QUESTION: Is it the case that if $\terma$ is a.s.~terminating then $\frac{d \oper{\terma}}{d \nu}$ exists? and if it exists, is it a.e.~differentiable?}

\section{Differentiability of the Weight and Value Functions}
\label{sec:diff}


To reason about the differential properties of these functions we place ourselves in a setting in which differentiation makes sense. We start with some preliminaries.

\subsection{Background on Differentiable Functions}
\label{sec:bdiff}

Basic real analysis gives a standard notion of differentiability at a point $\tupf x \in \Real^n$ for functions between Euclidean spaces $\Real^n \to \Real^m$.
In this context a function $f : \Real^n \to \Real^m$ is \defn{smooth} on an open $U \subseteq \Real^n$ if it has derivatives of all orders at every point of $U$.
The theory of \emph{differential geometry} \changed[lo]{(see e.g.~the textbooks 
\cite{Tu11,Lee13,Lee09})} abstracts away from Euclidean spaces to \emph{smooth manifolds}.
We recall the formal definitions.
\lo{Lang's book is not an ideal reference here: it treats Banach manifolds, rather than ordinary (smooth) manifolds.}

A topological space $\calM$ is \defn{locally Euclidean at a point $x \in \calM$} if $x$ has a neighbourhood $U$ such that there is a homeomorphism $\phi$ from $U$ onto an open subset of $\Real^n$, for some $n$. The pair $(U, \phi : U \to \Real^n)$ is called a \defn{chart} \changed[lo]{(of dimension $n$)}.
We say $\calM$ is \defn{locally Euclidean} if it is locally Euclidean at every point.
A \defn{manifold} $\calM$ is a Hausdorff, second countable, locally Euclidean space.

Two charts, $(U, \phi : U \to \Real^n)$ and $(V, \psi : V \to \Real^m)$, are \defn{compatible} if
\changed[lo]{the function
\(
\psi \circ \phi^{-1} : \phi(U \cap V) \to \psi(U \cap V)
\)
is smooth, with a smooth inverse}.
An \defn{atlas} on $\calM$ is a family $\{(U_\alpha, \phi_\alpha)\}$ of pairwise compatible charts that cover $\calM$.
A \defn{smooth manifold} is a manifold equipped with an atlas.

\changed[lo]{It follows from the topological invariance of dimension that charts that cover a part of the same connected component have the same dimension.
We emphasise that, although this might be considered slightly unusual, distinct connected components need not have the same dimension.}
This is important for our purposes: $\traces$ is easily seen to be a smooth manifold since each connected component $\traces_i$ is \changed[lo]{diffeomorphic} to $\Real^i$.
It is also straightforward to endow the set $\terms$ of SPCF terms with a (smooth) manifold structure.
Following \cite{DBLP:conf/icfp/BorgstromLGS16}
we view $\terms$ as $\bigcup_{m\in\Nat} \big(\skterms_m \times \Real^m \big)$,
where $\skterms_m$ is the set of SPCF terms with exactly $m$ place-holders (a.k.a.~\emph{skeleton terms}) for numerals.
Thus identified, we give $\terms$ the countable disjoint union topology of the product topology of the discrete topology on $\skterms_m$ and the standard topology on $\Real^m$.
Note that the connected components of $\terms$ have the form $\{M\} \times \Real^m$, with $M$ ranging over $\skterms_m$, and $m$ over $\Nat$. So in particular, the subspace $\values \subseteq \terms$ of values inherits the manifold structure. We fix the Borel algebra of this topology to be the $\sigma$-algebra on $\terms$.


\changed[lo]{Given manifolds $(\calM, \{U_\alpha, \phi_\alpha\})$ and $(\calM', \{V_\beta, \psi_\beta\})$,
a function $f: \calM \to \calM'$ is \defn{differentiable} at a point $x \in \calM$ if there are charts $(U_\alpha, \phi_\alpha)$ about $x$ and $(V_\beta, \psi_\beta)$ about $f(x)$ such that the composite $\psi_\beta \circ f \circ \phi_\alpha^{-1}$ restricted to the open subset $\phi_\alpha(f^{-1}(V_\beta) \cap U_\alpha)$ is differentiable at $\phi_\alpha(x)$.}

The definitions above are useful because they allow for a uniform presentation.
\changed[lo]{But it is helpful to unpack the definition of differentiability in a few instances, and we see that they boil down to the standard sense in real analysis.}
Take an SPCF term $M$ with free variables amongst $x_1, \ldots, x_m$ (all of type $\Real$), and $\args\in\Real^m\times\traces_n$.
\begin{itemize}
\item The function $\weightfn_M : \Real^m \times \traces \to \pReal$ is differentiable at $\args\in\Real^m\times\traces_n$ just if its restriction $\weightfn_M\vert_{\Real^m\times\traces_n}\from\Real^m\times\traces_n\to\pReal$ is differentiable at $\args$.
\item \changed[lo]{In case $M$ is of type $\PCFReal$, $\valuefn_M  : \Real^m \times \traces \to \closedvalues \cup \{\bot\}$ is in essence a partial function $\Real^m\times\traces\pto\Real$.
Precisely $\valuefn_M$ is differentiable at $\args$ just if for some open neighbourhood $U\subseteq\Real^m\times\traces_n$ of $\args$:}
\begin{enumerate}[noitemsep]
\item $\valuefn_M(\tupf{r'},\traceb)=\bot$ for all $(\tupf{r'},\traceb)\in U$; or
\item $\valuefn_M(\tupf{r'},\traceb)\neq\bot$ for all $(\tupf{r'},\traceb)\in U$, and $\valuefn_M'\from U\to\Real$ is differentiable at $\args$, where we define $\valuefn'_M(\tupf{r'},\tupf{s'})\defeq r''$ whenever $\valuefn_M(\tupf{r'},\tupf{s'})=\PCF{r''}$.
\end{enumerate}
\end{itemize}

\subsection{Why Almost Everywhere Differentiability Can Fail}
\label{subsec:on diff}

\hp{I changed the presentation of 4.2 a bit. The content is the same.}

Conditional statements break differentiability. This is easy to see with an example: the weight function of the term
\[
\If{\Sample \leq \Sample}{\Score{\PCF{1}}}{\Score{\PCF{0}}}
\]
is exactly the characteristic function of $\{[s_1,s_2] \in \traces \mid s_1 \leq s_2\}$,
which is not differentiable on the diagonal $\{[s,s] \in \traces_2 \mid s \in (0,1)\}$.

This function is however differentiable \emph{almost everywhere}: the diagonal is an uncountable set but has $\leb_2$ measure zero in the space $\traces_2$. Unfortunately, this is not true in general. Without sufficient restrictions, conditional statements also break almost everywhere differentiability. This can happen for two reasons.

\subsubsection{Problem 1: Pathological Primitive Functions.}
\label{par:primitive}
Recall that our definition of SPCF is parametrised by a set $\primitives$ of primitive functions. It is tempting in this context to take $\primitives$ to be the set of all differentiable functions, but this is too general, as we show now. Consider that for every $f : \Real \to \Real$ the term
\[
\Ifleq{\PCF{f}(\Sample)}{\Score{\PCF{1}}}{\Score{\PCF{0}}}
\]
has weight function the characteristic function of $\{\traceseq{s_1}\in\traces\mid f(s_1) \leq 0 \}$.
This function is non-differentiable at every $s_1\in\traces_1\cap\boundary f^{-1}(-\infty,0]$:
in every neighbourhood of $s_1$ there are $s_1'$ and $s_1''$ such that $f(s'_1)\leq 0$ and $f(s''_1)>0$.
One can construct a differentiable $f$ for which this is \emph{not} a measure zero set. (For example, there exists a non-negative function $f$ which is zero exactly on a \emph{fat} Cantor set, i.e., a Cantor-like set with strictly positive measure. See \cite[Ex.~5.21]{Rudin1976}.)

\subsubsection{Problem 2: Non-Terminating Runs.}
Our language has recursion, so we can construct a term which samples a random number, halts if this number is in $\mathbb Q \cap [0, 1]$, and diverges otherwise. In pseudo-code:
\begin{lstlisting}
  let rec enumQ p q r =
    if (r = p/q) then (score 1) else
      if (r < p/q) then
        enumQ p (q+1) r
      else
        enumQ (p+1) q r
  in enumQ 0 1 sample
\end{lstlisting}
The induced weight function is the characteristic function of $\{[s_1]\in \traces \mid s_1\in\Rational\}$;
the set of points at which this function is non-differentiable is $\traces_1$, which has measure $1$.

We proceed to overcome Problem 1 by making appropriate assumptions on the set of primitives.
We will then address Problem 2 by
focusing on \emph{almost surely terminating} programs.

\subsection{Admissible Primitive Functions}
One contribution of this work is to identify sufficient conditions for $\pop$. We will show in \Cref{subsec:main result} that our main result holds provided:

\begin{assumption}[Admissible Primitive Functions]
  \label{ass:pop}
  $\pop$ is a set of partial, measurable functions $\Real^\ell\pto\Real$ including all constant and projection functions which satisfies
  \begin{enumerate}
  \item \label{ass:comp} if $f: \Real^\ell \pto \Real$ and $g_i : \Real^{m} \pto \Real$ are elements of $\pop$ for $i = 1, \dots,\ell$, then $f \circ \langle g_i \rangle_{i = 1}^\ell : \Real^m \pto \Real$ is in $\pop$
  \item\label{ass:diff} if $(f: \Real^\ell \pto \Real) \in\pop$, then $f$ is differentiable in the interior of $\dom f$
  \item \label{ass:boundary} if $(f: \Real^\ell \pto \Real) \in\pop$, then $\leb_\ell(\boundary f^{-1}[0,\infty))=0$.
  \end{enumerate}
\end{assumption}

\begin{example}
\label{eg:primitive-functions}
The following sets of primitive operations satisfy the above sufficient conditions. (See \ifproceedings\cite{MOPW20} \else\Cref{apx:pop} \fi for a proof.)
\begin{enumerate}
\item The set $\pop_1$ of analytic functions with co-domain $\Real$.
Recall that a function $f:\Real^\ell \to \Real^n$ is \emph{analytic} if it is infinitely differentiable and its multivariate Taylor expansion at every point $x_0 \in \Real^\ell$ converges pointwise to $f$ in a neighbourhood of $x_0$.

\item
The set $\pop_2$ of (partial) functions $f\from\Real^\ell\pto\Real$ such that 
\changed[lo]{$\dom f$ is open\footnote{This requirement is crucial, and cannot be relaxed.}, and $f$ is differentiable everywhere in $\dom{f}$, and}
$f^{-1}(I)$ is a finite union of (possibly unbounded) rectangles\changed[dw]{\footnote{i.e. a finite union of
$I_1\times\cdots\times I_\ell$ for (possibly unbounded) intervals $I_i$}} for (possibly unbounded) intervals $I$.
\end{enumerate}
\end{example}

\lo{NOTE. For every rectangle $R$, $\leb(\boundary R) = 0$.
To see this, notice that
\[
\interior{R} =  (a_1, b_1) \times \cdots \times (a_n, b_n)
= \bigcup_{k=1}^\infty \prod_{i=1}^n[a_i + 1/k, b_i-1/k].
\]
Hence $\interior R$ is measurable, and $\leb(\interior{R}) = \leb(\overline R)$ (thanks to \cref{lem:increasing seq limit}).

\begin{lemma}[Limit]\label{lem:increasing seq limit}
If $(A_i)_{i \in \omega}$ is an increasing sequence of measurable sets (i.e.~$A_i \subseteq A_{i+1}$ for all $i \in \omega$) then $\mu(\cup_{i = 1}^\infty A_i) = \lim_{i \to \infty} \mu(A_i)$.
\end{lemma}}

Note that all primitive functions mentioned in our examples (and in particular the density of the normal distribution) are included in both $\pop_1$ and $\pop_2$.

It is worth noting that both $\pop_1$ and $\pop_2$ satisfy the following stronger (than \cref{ass:pop}.3) property:
$\leb_n(\boundary f^{-1}I)=0$ for every interval $I$, for every primitive function $f$.

\subsection{Almost Sure Termination}

\changed[dw]{To rule out the contrived counterexamples which diverge we restrict attention to}
\emph{almost surely terminating} SPCF terms.
Intuitively, a program $M$ (closed term of ground type) is almost surely terminating if the probability that a run of $M$ terminates is 1.

Take an SPCF term $M$ with variables amongst $x_1, \ldots, x_m$ (all of type $\Real$), and set
\begin{equation}
\trtermM \defeq \big\{
\args \in \Real^m \times \traces \mid \exists V, w \,.\, \config {M[{\tupf r} / {\tupf x}]}1\emptytrace\red^*\config V w\trace
\big\}.
\label{eq:trtermM}
\end{equation}
\changed[lo]{Let us first consider the case of closed $M \in \terms^0$ i.e.~$m = 0$
(notice that the measure $\mu_{\Real^m \times \traces}$ is not finite, for $m \geq 1$).
As $\trtermM$ now coincides with $\valuefn_M^{-1}(\closedvalues)$, $\trtermM$ is a measurable subset of $\traces$.}
Plainly if $M$ is deterministic (i.e.~$\Sample$-free), then $\trmeasure(\trtermM) = 1$ if $M$ converges to a value, and 0 otherwise.
Generally for an arbitrary (stochastic) term $M$ we can regard $\trmeasure(\trtermM)$ as the probability that a run of $M$ converges to a value, because of \cref{lem:termleq1}.

\begin{restatable}{lemma}{terminationsubprob}
\label{lem:termleq1}
If $M \in \terms^0$ then
\(
  \trmeasure(\trtermM)\leq 1.
\)
\end{restatable}

More generally, if $M$ has free variables amongst $x_1, \ldots, x_m$ (all of type $\PCFReal$), then we say that $M$ is almost surely terminating if for almost every (instantiation of the free variables by) ${\tupf r} \in \Real^m$, $M[{\tupf r} / {\tupf x}]$ terminates with probability 1.

We formalise the notion of almost sure termination as follows.

\begin{definition}\rm
  \label{def:ast}
  Let $M$ be an SPCF term.
  We say that $M$ \defn{terminates almost surely} if
  \begin{enumerate}[noitemsep]
  \item $\terma$ is closed and
  $\mu(\trtermM)=\mu(\valuefn_M^{-1}(\closedvalues))=1$; 
or
  \item $\terma$ has free variables amongst $x_1,\ldots,x_m$ (all of which are of type $\PCFReal$), and there exists $T \in \Sigma_{\Real^m}$ such that $\leb_m(\Real^m\setminus T)=0$ and for each $\tupf r\in T$, $M[\tupf{\PCF r}/\tupf x]$ terminates almost surely.
  \end{enumerate}
\end{definition}

Suppose that $M$ is a closed term and \changed[lo]{$M^\flat$} is obtained from $M$ by recursively replacing subterms $\Score L$ with the term $\If {L<0}{N_\Fail}{L}$, where $N_\Fail$ is a term that reduces to $\Fail$ such as $\PCF 1/\PCF 0$.
It is easy to see that for all $\tupf s\in\traces$, $\config {M^\flat} 1 \emptytrace\red^*\config V 1{\tupf s}$ iff for some (unique) $w\in\pReal$, $\config M 1 \emptytrace\red^*\config V w {\tupf s}$. Therefore,
  \begin{align*}
    \llbracket M^\flat\rrbracket(\values)&=\int_{\valuefn^{-1}_{M^\flat}(\values)}  \diff \trmeasure \;\weightfn_{M^\flat}\\
                                    &=\trmeasure(\{\tupf s\in\traces\mid\exists V \, . \, \config{M^\flat}1\emptytrace\red^*\config V 1{\tupf s}\})
                                    =\trmeasure(\trtermM
                                      )
  \end{align*}
Consequently, the closed term $M$ terminates almost surely iff $\llbracket M^\flat\rrbracket$ is a probability measure.

\begin{remark}
\begin{itemize}
\item Like many treatments of semantics of probabilistic programs in the literature, we make no distinction between non-terminating runs and aborted runs of a (closed) term $M$: both could result in the value semantics $\llbracket M^\flat\rrbracket$ being a sub-probabilty measure 
(cf.\ \cite{DBLP:conf/esop/BichselGV18}).


\item
Even so, current probabilistic programming systems do not place any restrictions on the code that users can write:
it is perfectly possible to construct invalid models because catching programs that do not define valid probability distributions can be hard, or even impossible.
This is not surprising, because almost sure termination is hard to decide: it is $\Pi_2^0$-complete in the arithmetic hierarchy \cite{DBLP:journals/acta/KaminskiKM19}.
Nevertheless, because a.s.~termination is an important correctness property of probabilistic programs (not least because of the main result of this paper, \cref{thm:mainres}), the development of methods to prove a.s.~termination is a hot research topic.
\end{itemize}
\end{remark}







Accordingly the main theorem of this paper is stated as follows:
\begin{restatable*}{theorem}{mainres}
  \label{thm:mainres}
  Let $M$ be an SPCF term (possibly with free variables of type $\PCFReal$) which terminates almost surely.
  Then its weight function $\weightfn_M$ and value function $\valuefn_M$ are differentiable almost everywhere.
\end{restatable*}

\section{Stochastic Symbolic Execution}
\label{sec:symbolic}




\dw{todo: State somewhere explicitly that we keep track of the set of traces rather than conditions on the traces. This makes it look closer to the probabilistic operational semantics but less similar to standard symbolic execution.}

We have seen that a source of discontinuity is the use of if-statements. Our main result therefore relies on an in-depth understanding of the branching behaviour of programs.
The operational semantics given in \cref{sec:semantics}
is unsatisfactory in this respect: any two execution paths
are treated independently, whether they go through different branches of an if-statement or one is obtained from the other by
using slightly perturbed random samples not affecting the control flow.

More concretely, note that although we have derived $\weightfn_{\sf Ped}\traceseq{0.2,0.9,0.7}=0.54$ and $\valuefn_{\sf Ped}\traceseq{0.2,0.9,0.7}=\PCF{0.6}$ in \cref{example: reduction}, we cannot infer anything about $\weightfn_{\sf Ped}\traceseq{0.21,0.91,0.71}$ and $\valuefn_{\sf Ped}\traceseq{0.21,0.91,0.71}$ unless we perform the corresponding reduction.

So we propose an alternative
\emph{symbolic} operational semantics (similar to the ``compilation scheme'' in \cite{DBLP:conf/aistats/ZhouGKRYW19}), in which no sampling is performed: whenever a $\Sample$ command is encountered,
we simply substitute a fresh variable $\alpha_i$ for it, and continue on with the execution.
We can view this style of semantics as a stochastic form of symbolic execution \cite{DBLP:journals/tse/Clarke76,DBLP:journals/cacm/King76}, i.e., a means of analysing a program so as to determine what \emph{inputs}, and \emph{random draws} (from $\Sample$) cause each part of a program to execute.

Consider the term $ M \defeqs
  \letin{x =\Sample\cdot\PCF 3}{(\walk\,x)}$, defined using the function $\passZero$ of \cref{example: SPCF syntax}.
  We have a reduction path
\begin{align*}
  M
  \sred
  \letin{( x =\alpha_1\cdot\PCF 3)}{(\walk\,x)}
  \sred
  \walk\,(\alpha_1\cdot \PCF{3})
\end{align*}
but at this point we are stuck: the CBV strategy requires a value for $\alpha_1$. We will ``delay'' the evaluation of the multiplication $\alpha_1\cdot\PCF 3$; we signal this by drawing a box around the delayed operation:  $\alpha_1\,\delay\cdot\,\PCF 3$. We continue the execution, inspecting the definition of $\passZero$, and get:
\begin{align*}
  M
  & \sred^*
  \passZero\,(\alpha_1\,\delay\cdot\, \PCF{3})
  \sred^*
  N \defeqs \Ifleq{\alpha_1\,\delay\cdot\,\PCF{3}}{\PCF{0}}{P}
\end{align*}
where
 $$P \defeqs \begin{pmatrix*}[l]
                      \letin{s = \Sample} \\
                      \If{(\Sample \leq \PCF{0.5})}
                      {\big(s + \passZero(\alpha_1\,\delay\cdot\,\PCF{3}+s)\big)}{\big(s + \passZero(\alpha_1\,\delay\cdot\,\PCF{3}-s)\big)}
                    \end{pmatrix*}.$$
We are stuck again: the value of $\alpha_1$ is needed in order to know which branch to follow. Our approach consists in considering the space $\traces_1 = (0, 1)$ of possible values for $\alpha_1$, and splitting it into $\{s_1\in(0,1)\mid s_1\cdot 3 \leq 0\}=\emptyset$
and
$\{s_1\in(0,1)\mid s_1\cdot 3> 0\}=(0,1)$.
Each of the two branches will then yield a weight function restricted to the appropriate subspace.

Formally,
our symbolic operational semantics is a rewrite system of configurations of the form $\sconfig\sterm\sweight U$,
where $\sterm$ is a term with delayed (boxed) operations, and free ``sampling'' variables\footnote{Note that $\sterm$ may be open and contain other free ``non-sampling'' variables, usually denoted $x_1, \dots, x_m$.} $\alpha_1,\ldots,\alpha_n$; $U\subseteq\traces_n$ is the subspace of sampling values compatible with the current branch; and $\sweight\from U\to\pReal$ is a function assigning to each $\tupf s\in U$ a weight $\sweight(\tupf s)$. In particular, for our running example\footnote{We use the meta-lambda-abstraction $\bblambda x\ldotp f(x)$ to denote the set-theoretic function $x \mapsto f(x)$.} 
\begin{align*}
\sconfig M{\bblambda\emptytrace\ldotp 1}{\traces_0}&\sred^*\sconfig N{\bblambda\traceseq{s_1}\ldotp 1}
{(0, 1)}.
\end{align*}
As explained above, this leads to two branches:
\[
\begin{tikzcd}[row sep=-0.8em,column sep=1.2em]
&  \sconfig {\PCF 0}{\bblambda\traceseq{s_1}\ldotp 1}{\emptyset} \phantom{aaa} \\
\sconfig N{\bblambda\traceseq{s_1}\ldotp 1}{(0,1)} \arrow[phantom]{ur}[description]{\raisebox{1em}{\rotatebox[origin=c]{+25}{$\sred^*$}}}
\arrow[phantom]{dr}{\rotatebox[origin=c]{-25}{$\sred^*$}}
&  \\
& \sconfig P {\bblambda\traceseq{s_1}\ldotp 1}{(0,1)}
\end{tikzcd}
\]
The first branch has reached a value, and the reader can check that the second branch continues as
{\small
\begin{align*}
& \sconfig P {\bblambda\traceseq{s_1}\ldotp 1}{(0,1)} \sred^*  \\
& \
{\sconfig
        {
          \If
            {\alpha_3\leq\PCF{0.5}}
            {\alpha_2+\passZero(\alpha_1\,\delay\cdot\,\PCF 3+\alpha_2)}
            {\alpha_2+\passZero(\alpha_1\,\delay\cdot\,\PCF 3-\alpha_2)}
        }
        {\bblambda\traceseq{s_1,s_2,s_3}\ldotp 1}
        {(0,1)^3}
}
\end{align*}
}
where $\alpha_2$ and $\alpha_3$ stand for the two $\Sample$ statements in $P$. From here we proceed by splitting $(0,1)^3$ into
$(0,1)\times(0,1)\times (0,0.5]$ and $(0,1)\times(0,1)\times(0.5,1)$
\changed[dw]{and after having branched again (on whether we have passed $0$) the evaluation of $\walk$ can terminate in the configuration
\begin{align*}
  \sconfig {\alpha_2\,\delay +\, 0}{\bblambda\traceseq{s_1,s_2,s_3}\ldotp 1}{U}
\end{align*}
where $U\defeq\{\traceseq{s_1,s_2,s_3}\in\traces_3\mid s_3> 0.5\land s_1\cdot 3-s_2\leq 0\}$.}


Recall that $M$ appears in the context of our running example $\sf Ped$.
Using our calculations above we derive one of its branches:
\begin{align*}
&\sconfig{\sf Ped}{\bblambda\emptytrace\ldotp 1}{\{\emptytrace\}}
\sred^*\sconfig{
  \letin{w = \Score{\PCF{\pdf_{\Normal{1.1}{0.1}}}(\alpha_2)}}
  {\alpha_1\,\delay\cdot\,\PCF 3}
}{\bblambda\traceseq{s_1,s_2,s_3}\ldotp 1}U \\
  &\hspace*{7mm}\sred\sconfig{
    \letin{w = \Score{\delay{\pdf_{\Normal{1.1}{0.1}}}(\alpha_2)}}
    {\alpha_1\,\delay\cdot\,\PCF 3}
  }{\bblambda\traceseq{s_1,s_2,s_3}\ldotp 1}U\\
  &\hspace*{7mm}\sred^*\sconfig{
    \letin{w = \delay{\pdf_{\Normal{1.1}{0.1}}}(\alpha_2)}
    {\alpha_1\,\delay\cdot\,\PCF 3}
  }{\bblambda\traceseq{s_1,s_2,s_3}\ldotp\pdf_{\Normal{1.1}{0.1}}(s_2)}U\\
  &\hspace*{7mm}\sred^*\sconfig{\alpha_1\,\delay\cdot\,\PCF 3}{\bblambda\traceseq{s_1,s_2,s_3}\ldotp\pdf_{\Normal{1.1}{0.1}}(s_2)}U
\end{align*}
In particular the trace $\traceseq{0.2,0.9,0.7}$ of \cref{example: reduction} lies in the subspace $U$. We can immediately read off the corresponding value and weight functions for \emph{all} $\traceseq{s_1,s_2,s_3}\in U$ simply by evaluating the computation $\alpha_1\cdot\PCF 3$, which we have delayed until now:
\begin{align*}
  \valuefn_{\sf Ped}\traceseq{s_1,s_2,s_3}&=\PCF{s_1\cdot 3}&
  \weightfn_{\sf Ped}\traceseq{s_1,s_2,s_3}&=\pdf_{\Normal{1.1}{0.1}}(s_2)
\end{align*}


\subsection{Symbolic Terms and Values}
\label{sec:stermval}



We have just described informally our symbolic execution approach, which involves delaying the evaluation of primitive operations.
We make this formal by introducing an extended notion of terms, which we call \defn{symbolic terms} and define in
\cref{fig:sterm} along with a notion of \defn{symbolic values}. For this we assume fixed denumerable sequences of \defn{distinguished} variables: $\alpha_1, \alpha_2, \ldots$, used to represent sampling, and $x_1,x_2,\ldots$ used for free variables of type $\PCFReal$.
Symbolic terms are typically denoted $\sterm$, $\stermb$, or $\stermc$.
They contain terms of the form $\delay f(\sval_1,\ldots,\sval_\ell)$ for $f\from\Real^\ell\pto\Real\in\pop$ a primitive function, representing delayed evaluations, and they also contain the sampling variables $\alpha_j$.
The type system is adapted in a straightforward way, see \cref{fig:styping}.

We use $\termss m n$ to refer to the set of well-typed symbolic terms with free variables amongst $x_1,\ldots,x_m$ and $\alpha_1,\ldots,\alpha_n$ (and all are of type $\PCFReal$).
Note that every term in the sense of \cref{fig:SPCF syntax} is also a symbolic term. 
\lo{@Dominik: Why do we highlight certain expressions in \cref{fig:symsem}? I can't remember your answer.}

\dw{The aim is to draw the reader's attention to the most noteworthy items. In \cref{fig:sterm} it is e.g.\ what is added in comparison to regular terms.}

\lo{OK. I have added a footnote at the first reference to \cref{fig:SPCF syntax} to explain the highlighting.}

\begin{figure}[ht]
  \begin{subfigure}{\textwidth}
    \begin{align*}
      \sval &::=\PCF{r}\mid x_i\mid \mathhighlight{\alpha_j } \mid \mathhighlight{\delay f(\sval_1,\ldots,\sval_\ell) } \mid\lambda y\ldotp\sterm \\
      \sterm,\stermb,\stermc&::= \sval \mid y  \mid
                              \PCF{f}(\sterm_1,\dots,\sterm_\ell)  \mid
                              \sterm\,\stermb \mid
                              \Y{\sterm}
                            \mid\If{\stermc \leq 0}{\sterm}{\stermb}
                              \mid
                              \Sample \mid
                              \Score{\sterm}
    \end{align*}
    \caption{Symbolic values (typically $\sval$) and symbolic terms (typically $\sterm$, $\stermb$ or $\stermc$)}
    \label{fig:sterm}
  \end{subfigure}
  \begin{subfigure}{\textwidth}
    $$\begin{array}{c}
      \mathhighlight{\infer{\Gamma\vdash \delay f(\sval_1,\ldots,\sval_\ell)\from\PCFReal}{\Gamma\vdash\sval_1\from\PCFReal\cdots\Gamma\vdash\sval_\ell\from\PCFReal}}\qquad\infer{\Gamma\vdash x_i\from\PCFReal}{}\qquad \mathhighlight{\infer{\Gamma\vdash \alpha_j\from\PCFReal}{}}\\[1em]
      \infer{\Gamma,y\from\typea\vdash y\from\typea}{}\qquad \infer[r\in\Real]{\Gamma\vdash\PCF r\from\PCFReal}{}\qquad
            \infer{\Gamma\vdash\PCF f(\sterm_1,\ldots,\sterm_\ell)\from\PCFReal}{\Gamma\vdash\sterm_1\from\PCFReal\cdots\Gamma\vdash\sterm_\ell\from\PCFReal}\\[1em]
      \infer{\Gamma\vdash\lambda y\ldotp\sterm\from\typea\to\typeb}{\Gamma,y\from\typea\vdash\sterm\from\typeb}\qquad \infer{\Gamma\vdash\sterm\,\stermb\from\typeb}{\Gamma\vdash\sterm\from\typea\to\typeb\Gamma\vdash\stermb\from\typea}\qquad
      \infer{\Gamma\vdash\Y\sterm\from\typea\tyarrow\typeb}{\Gamma\vdash\sterm\from(\typea\tyarrow\typeb)\tyarrow\typea\tyarrow\typeb}\\[1em]
      \infer{\Gamma\vdash\Ifleq\stermc\sterm\stermb\from\typea}{\Gamma\vdash\stermc\from\PCFReal&\Gamma\vdash\sterm\from\typea\Gamma\vdash\stermb\from\typea}\qquad
      \infer{\Gamma\vdash\Sample\from\PCFReal}{}\qquad
      \infer{\Gamma\vdash\Score\sterm\from\PCFReal}{\Gamma\vdash\sterm\from\PCFReal}
    \end{array}$$
    \caption{Type system for symbolic terms}
    \label{fig:styping}
  \end{subfigure}
  \begin{subfigure}{\textwidth}
    \begin{align*}
      \sredex&::=(\lambda y\ldotp\sterm)\,\sval\mid \PCF f(\mathhighlight{\sval_1,\ldots,\sval_\ell})\mid\Y{(\lambda y\ldotp\sterm)}\mid\Ifleq{\mathhighlight{\sval}}\sterm\stermb\mid\Sample\mid\Score{\mathhighlight\sval}\\
      \scon&::=[\,]\mid\scon\,\sterm\mid (\lambda y\ldotp\sterm)\, \scon\mid\PCF f(\mathhighlight{\sval_1,\ldots, \sval_{i-1}},\scon,\sterm_{i+1},\ldots,\sterm_\ell)\mid\Y\scon\mid\\
             &\Ifleq{\scon}\sterm\stermb\mid \Score{\scon}
    \end{align*}
    \caption{Symbolic values (typically $\sval$), redexes ($\sredex$) and  reduction contexts ($\scon$).}
    \label{fig:sconred}
  \end{subfigure}
  \caption{Symbolic terms and values, type system, reduction contexts, and redexes. As usual $f\in\pop$ and $r\in\Real$.}
  \label{fig:symsem}
\end{figure}

\begin{figure}[ht]
  \begin{subfigure}{\textwidth}
    \begin{align*}
      \domnp\conc{\delay f(\sval_1,\ldots,\sval_\ell)}&\defeq\{(\tupf r,\tupf s)\in\domnp\conc{\sval_1}\cap\cdots\cap\domnp\conc{\sval_\ell}\mid(r'_1,\ldots,r'_\ell)\in\dom f,\\
      &\hspace{28mm} \text{where } \PCF{r'_1}=\conc{\sval_1}(\tupf r,\tupf s), \cdots, \PCF{r'_\ell}=\conc{\sval_\ell}(\tupf r,\tupf s) \}\\
      \domnp\conc\Sample&\defeq\domnp\conc{x_i}\defeq\domnp\conc{\alpha_j}\defeq\domnp\conc y\defeq\domnp\conc{\PCF{r'}}\defeq\Real^m\times\traces_n\\
      \domnp{\PCF f(\sterm_1,\ldots,\sterm_\ell)}&\defeq\domnp\conc{\sterm_1}\cap\cdots\cap\domnp\conc{\sterm_\ell}\\
      \domnp\conc{\lambda y\ldotp\sterm}&\defeq\domnp\conc{\Y\sterm}\defeq\domnp\conc{\Score\sterm}\defeq\domnp\conc\sterm\\
      \domnp\conc{\sterm\,\stermb}&\defeq\domnp\conc\sterm\cap\domnp\conc\stermb\\
      \domnp\conc{\Ifleq\stermc\sterm\stermb}&\defeq\domnp\conc\stermc\cap\domnp\conc\sterm\cap\domnp\conc\stermb
    \end{align*}
    \caption{Domain of $\conc\cdot$}
    \label{fig:dominst}
  \end{subfigure}
  \begin{subfigure}{\textwidth}
    \begin{align*}
      \conc{\delay f(\sval_1,\ldots,\sval_\ell)}(\tupf r,\tupf s)&\defeq \mathhighlight{\PCF {f(r'_1,\ldots, r'_\ell)}}\text{, where for $1\leq i\leq\ell$, }\conc{\sval_i}(\tupf r,\tupf s)=\PCF{ r'_i}\\
      \conc{x_i}(\tupf r,\tupf s)&\defeq\PCF{ r_i}\\
      \conc{\alpha_j}(\tupf r,\tupf s)&\defeq\PCF{ s_j}\\
      \conc y (\tupf r,\tupf s)&\defeq y \\
      \conc {\PCF r'}(\tupf r,\tupf s)&\defeq\PCF {r'}\\
      \conc{\PCF f(\sterm_1,\ldots,\sterm_\ell)}(\tupf r,\tupf s)&\defeq \mathhighlight{\PCF f(\conc{\sterm_1}(\tupf r,\tupf s),\ldots,\conc{\sterm_\ell}(\tupf r,\tupf s))}\\
      \conc{\lambda y\ldotp\sterm}(\tupf r,\tupf s)&\defeq\lambda y\ldotp\conc\sterm(\tupf r,\tupf s) \\
      \conc{\sterm\,\stermb}(\tupf r,\tupf s)&\defeq (\conc\sterm(\tupf r,\tupf s))\,(\conc\stermb(\tupf r,\tupf s))\\
      \conc{\Y\sterm}(\tupf r,\tupf s)&\defeq\Y{(\conc\sterm(\tupf r,\tupf s))}\\
      \conc{\Ifleq\stermc\sterm\stermb}(\tupf r,\tupf s)&\defeq\Ifleq{\conc\stermc(\tupf r,\tupf s)}{\conc\sterm(\tupf r,\tupf s)}{\conc\stermb(\tupf r,\tupf s)}\\
      \conc\Sample (\tupf r,\tupf s)&\defeq\Sample\\
      \conc{\Score\sterm}(\tupf r,\tupf s)&\defeq\Score{\conc\sterm(\tupf r,\tupf s)}
    \end{align*}
    \caption{Definition of $\conc\cdot$ on $\domnp\conc\cdot$}
    \label{fig:definst}

  \end{subfigure}
  \caption{Formal definition of the instantiation and partial evaluation function $\conc\cdot$}
  \label{fig:inst}
\end{figure}

Each symbolic term $\sterm\in\termss m n$ has a corresponding set of regular terms, accounting for all possible values for its sampling variables $\alpha_1, \dots, \alpha_n$ 
and its (other) free variables $x_1, \dots, x_m$. For $\tupf r\in\Real^m$ and \changed[lo]{$\trace\in \traces_n$}, we call \defn{partially evaluated instantiation} of $\sterm$ the term $\conc\sterm(\tupf r,\tupf s)$ obtained from $\sterm[\tupf{\PCF r}/\tupf x,\tupf{\PCF s}/\tupf\alpha]$
by recursively ``evaluating'' subterms of the form $\delay f(\PCF{r_1},\ldots,\PCF{r_\ell})$ to $\PCF{f(r_1,\ldots,r_\ell)}$, provided $(r_1,\ldots,r_\ell)\in\dom f$. In this operation, subterms of the form $\PCF f(\PCF{r_1},\ldots,\PCF{r_\ell})$ are left unchanged, and so are any other redexes. \changed[dw]{$\conc\sterm$ can be viewed as a partial function $\conc\sterm\from\Real^m\times\traces_n\pto\terms$ and a formal definition is presented in \cref{fig:definst}. (To be completely rigorous, we define for \emph{fixed} $m$ and $n$, partial functions ${\conc\sterm}_{m,n}\from\Real^m\times\traces_n\pto\terms$ for symbolic terms $\sterm$ whose distinguished variables are amongst $x_1,\ldots,x_m$ and $\alpha_1,\ldots,\alpha_n$. $\sterm$ may contain other variables $y, z, \dots$ of any type. Since $m$ and $n$ are usually clear from the context, we omit them.)
Observe that for $\sterm\in\termss m n$ and $\args\in\domnp\conc\sterm$, $\conc\sterm\args$ is a closed term.
}


\begin{example}
  Consider $\sterm\defeqs (\lambda\sco\ldotp\alpha_1\,\delay\cdot\,\PCF 3)\,(\Score{\PCF{\pdf_{\Normal{1.1}{0.1}}}(\alpha_2)})$. Then, for $\tupf{r} = \emptytrace$ and $\tupf{s} = \traceseq{0.2,0.9,0.7}$, we have
$
   \conc\sterm \args =(\lambda\sco\ldotp\PCF {0.6})\,(\Score{\PCF{\pdf_{\Normal{1.1}{0.1}}}(\PCF{0.9})}).
$
\end{example}
More generally, observe that if $\Gamma\vdash\sterm\from\typea$ and $\args\in\domnp\conc\sterm$ then $\Gamma\vdash\conc\sterm\args\from\typea$.
In order to evaluate conditionals $\Ifleq\stermc\sterm\stermb$ we need to reduce $\stermc$ to a real constant, i.e., we need to have $\conc\stermc(\tupf r,\tupf s) = \PCF r$ for some $r \in \reals$. This is the case whenever $\stermc$ is a symbolic value of type $\PCFReal$, since these are built only out of delayed operations, real constants and distinguished variables $x_i$ or $\alpha_j$. Indeed we can show the following:

\changed[dw]{
  \begin{restatable}{lemma}{instval}
    \label{lem:instval}
    Let $\args\in\domnp\conc\sterm$.
    Then $\sterm$ is a symbolic value iff $\conc\sterm\args$ is a value.
  \end{restatable}
}



For symbolic values $\sval \from \PCFReal$ and $(\tupf r,\tupf s)\in\domnp\conc\sval$ we employ the notation $\seva{\sval}\args := r'$ provided that $\conc{\sval}\args = \PCF{r'}$.

A simple induction on symbolic terms and values yields the following property, which is crucial for the proof of our main result (\cref{thm:mainres}):
\begin{restatable}{lemma}{technicalprimitives}
  \label{lem:seva}
  Suppose the set $\pop$ of primitives satisfies \changed[dw]{\Cref{ass:comp} of \cref{ass:pop}}.
  \begin{enumerate}[noitemsep]
  \item For each symbolic value $\sval$ of type $\PCFReal$, by identifying $\domnp\seva\sval$ with a subset of $\Real^{m+n}$, we have $\seva\sval\in\pop$. 
  \item \changed[dw]{If $\pop$ also satisfies \cref{ass:diff} of \cref{ass:pop} then} for each symbolic term $\sterm$, $\conc\sterm\from\Real^m\times\traces_n\pto\terms$ is differentiable in the interior of its domain.
  \end{enumerate}
\end{restatable}

\subsection{Symbolic Operational Semantics}
\label{sec:sred}


We aim to develop a symbolic operational semantics that provides a sound and complete abstraction of the (concrete) operational trace semantics.
The symbolic semantics is presented as a rewrite system of \defn{symbolic configurations}, which are defined to be triples of the form $\sconfig\sterm\sweight U$, where for some $m$ and $n$, $\sterm\in\termss m n$, $U\subseteq\domnp\conc\sterm\subseteq\Real^m\times\traces_n$ is measurable, and $\sweight\from\Real^m\times\traces\pto\pReal$ with $\dom \sweight=U$.
\changed[lo]{Thus we aim to prove the following result (writing $\const$ for the constant function $\bblambda(\tupf r,\tupf s)\ldotp 1$):}


\changed[lo]{
\begin{theorem}
\label{thm:soundness and completeness}
Let $M$ be a term with free variables amongst $x_1,\ldots,x_m$.
\begin{enumerate}
  \item\label{it:sound} \textsf{(Soundness)}. If $\sconfig M\const {\Real^m}\sred^*\sconfig\sval\sweight U$
  then for all $(\tupf r,\tupf s)\in U$ it holds
    $\weightfn_M(\tupf r, \tupf s)=\sweight(\tupf r,\tupf s)$ and
    $\valuefn_M(\tupf r,\tupf s)=\conc\sval(\tupf r,\tupf s)$.
  \item\label{it:compl} \textsf{(Completeness)}. If $\tupf r\in\Real^m$ and $\config {M[\tupf {\PCF r}/\tupf x]} 1\emptytrace\red^*\config V w {\tupf s}$ then there exists 
    $\sconfig M \const{\Real^m}\sred^*\sconfig\sval\sweight U$ such that $(\tupf r,\tupf s)\in U$. 
  \end{enumerate}
\end{theorem}
}

\changed[lo]{As formalised by \cref{thm:soundness and completeness}, the key intuition behind symbolic configurations $\sconfig\sterm\sweight U$ (that are reachable from a given $\sconfig M\const {\Real^m}$) is that, whenever $\sterm$ is a symbolic value:
\begin{itemize}
\item $\sterm$ gives a correct \emph{local} view of $\valuefn_M$ (restricted to $U$), and
\item $\sweight$ gives a correct \emph{local} view of $\weightfn_M$ (restricted to $U$);
\end{itemize}
moreover, the respective third components $U$ (of the symbolic configurations $\sconfig\sterm\sweight U$) cover $\trtermM$.}

\dw{This isn't quite right. A way to fix it is to restrict this to symbolic values $\sval$.}

To establish \cref{thm:soundness and completeness}, we introduce \defn{symbolic reduction contexts} and \defn{symbolic redexes}.
These are presented in \cref{fig:sconred} and extend the usual notions \changed[dw]{(replacing real constants with arbitrary symbolic values of type $\PCFReal$)}.

Using \cref{lem:instval} we obtain:
\begin{lemma}
  \label{lem:instred} If $\sredex$ is a symbolic redex and $\args\in\domnp\conc\sredex$ then $\conc\sredex\args$ is a redex.
\end{lemma}

The following can be proven by a straightforward induction\ifproceedings\else (see \cref{app:sred})\fi:
\begin{restatable}[Subject Construction]{lemma}{ssubcon}
  Let $\sterm$ be a symbolic term.
  \label{lem:basic0}
  \begin{enumerate}
  \item If $\sterm$ is a symbolic value then for all symbolic contexts $\scon$ and symbolic redexes $ \sredex$, $\sterm\not\equiv\scon[\sredex]$.
  \item If $\sterm\equiv\scon_1[\sredex_1]\equiv\scon_2[\sredex_2]$ then $\scon_1\equiv\scon_2$ and $\sredex_1\equiv\sredex_2$.
  \item If $\sterm$ is not a symbolic value and $\domnp\conc\sterm\neq\emptyset$ then there exist $\scon$ and $\sredex$ such that $\sterm\equiv\scon[\sredex]$.
  \end{enumerate}
\end{restatable}

The partial instantiation function also extends to symbolic contexts $\scon$ in the evident way -- we give the full definition in \ifproceedings\cite{MOPW20}\else\cref{app:sred} (\cref{def:conccon})\fi.


Now, we introduce the following rules for \defn{symbolic redex contractions}:
\begin{align*}
  \sconfig{(\lambda y\ldotp \sterm)\,\sval}\sweight U&\sred\sconfig{\sterm[\sval/y]}\sweight U\\
  \sconfig{\PCF{f}(\sval_1,\ldots,\sval_\ell)}\sweight U&\sred\sconfig{\delay f(\sval_1,\ldots,\sval_\ell)}\sweight {\mathhighlight{\domnp\seva{\delay f(\sval_1,\ldots,\sval_\ell)}\cap U}} \\
  \sconfig{\Y (\lambda y\ldotp \sterm)}\sweight U&\sred\sconfig{\lambda z\ldotp \sterm\,[\Y (\lambda y\ldotp \sterm) / y] \, z}\sweight U\\
  \sconfig{\Ifleq\sval{\sterm}{\stermb}}\sweight U&\sred\sconfig{\sterm}\sweight {\mathhighlight{\seva\sval^{-1}(-\infty,0]\cap U}}\\
  \sconfig{\Ifleq\sval{\sterm}{\stermb}}\sweight U&\sred\sconfig{\stermb}\sweight {\mathhighlight{\seva\sval^{-1}(0,\infty)\cap U}}\\
  \sconfig{\Sample}\sweight U&\sred\sconfig{\mathhighlight{\alpha_{n+1}}}{\sweight'}{U'}
  \tag{$U\subseteq\Real^m\times\traces_n$}\\
  \sconfig{\Score{\sval}}\sweight U&\sred\sconfig{\sval}{\mathhighlight{\seva\sval\cdot\sweight}} {\mathhighlight{\seva\sval^{-1}[0,\infty)\cap U}}
\end{align*}
In the rule for $\Sample$, $U'\defeq\{(\tupf r,\tupf s\concat\traceseq{s'})\mid (\tupf r,\tupf s)\in U\land s'\in (0,1)\}$ and $\sweight'(\tupf r,\tupf s\concat\traceseq{s'})\defeq\sweight(\tupf r,\tupf s)$; in the rule for $\Score\sval$,  $(\seva\sval\cdot\sweight)\args\defeq\seva\sval\args\cdot\sweight\args$.

The rules are designed to closely mirror their concrete counterparts.
Crucially, the rule for $\Sample$ introduces a ``fresh'' sampling variable, and the two rules for conditionals split the last component $U\subseteq\Real^m\times\traces_n$ according to whether $\seva\sval(\tupf r,\tupf s)\leq 0$ or $\seva\sval(\tupf r,\tupf s)>0$.
The ``delay'' contraction (second rule) is introduced for a technical reason:
ultimately, to enable \cref{it:sound} \textsf{(Soundness)}.
Otherwise it is, for example, unclear whether $\lambda y\ldotp\alpha_1+\PCF 1$ should correspond to $\lambda y\ldotp\PCF{0.5}+\PCF 1$ or $\lambda y\ldotp\PCF{1.5}$ for $s_1=0.5$.





\dw{N.B.\ restrictions of $\sweight$ in the branching-rules are not made explicit}

Finally we lift this to arbitrary symbolic terms using the obvious rule for symbolic evaluation contexts:
\[
\frac{
{\sconfig \sredex \sweight U} \sred {\sconfig {\scontra} {\sweight'} {U'}}
}
{
\sconfig {\scon[\sredex]} \sweight U \sred {\sconfig {\scon[\scontra]} {\sweight'} {U'}}
}
\]




\changed[dw]{Note that we do not need rules corresponding to reductions to $\Fail$ because the third component of the symbolic configurations ``filters out'' the pairs $\args$ corresponding to undefined behaviour. In particular, the following holds:}
\begin{restatable}{lemma}{reductionvalidconfigs}
  \label{lem:pres}
  Suppose $\sconfig\sterm\sweight U$ is a symbolic configuration and $\sconfig\sterm\sweight U\sred\sconfig\stermb{\sweight'}{U'}$. Then $\sconfig\stermb{\sweight'}{U'}$ is a symbolic configuration.
\end{restatable}

A key advantage of the symbolic execution is that the induced computation tree is finitely branching, since branching only arises from conditionals, splitting the trace space into disjoint subsets.
This contrasts with the concrete situation (from \cref{sec:semantics}), in which sampling creates  uncountably many branches.
\begin{restatable}[Basic Properties]{lemma}{basicproperties}
  \label{lem:basic}
  Let $\sconfig\sterm\sweight U$ be a symbolic configuration. Then
  \begin{enumerate}[noitemsep]
  \item There are at most countably distinct such $U'$ that $\sconfig\sterm\sweight U\sred^*\sconfig{\stermb}{\sweight'}{U'}$.
  \item\label{it:disj} If $\sconfig\sterm\sweight U\sred^*\sconfig{\sval_i}{\sweight_i}{U_i}$ for $i\in\{1,2\}$ then $U_1=U_2$ or $U_1\cap U_2=\emptyset$.
  \item\label{it:disj2} If $\sconfig\sterm\sweight U\sred^*\sconfig{\scon_i[\Sample]}{\sweight_i}{U_i}$ for $i\in\{1,2\}$ then $U_1=U_2$ or $U_1\cap U_2=\emptyset$.
  \end{enumerate}
\end{restatable}

\lo{From Reviewer 1: Lemma 9 (Basic Properties) seems false. The antecedent ``If...$i\in\{1,2\}$'' is symmetric between $i=1$ and $i=2$, but the $\Rightarrow^*$ in part 2 is directed. So, if $\mathcal{M}_1$ is a value and $\mathcal{M}=\mathcal{M}_2=(\lambda y.y)(\mathcal{M}_1)$ and $U=U_1=U_2\ne\emptyset$, then $\mathcal{M}_2$ reduces to $\mathcal{M}_1$ but not vice versa, so part 2 seems to fail. The sentence about ``at most $2^k$ distinct $U'$'' makes sense, but how is it used below?}

\dw{I think we fixed this, didn't we?

  Regarding the $2^k$: In particular, it ensures that there are only countably many $U'$s, which is crucial later on.}

Crucially, there is a correspondence between the concrete and symbolic semantics in that they can ``simulate'' 
each other:
\begin{restatable}[Correspondence]{proposition}{correspondenceprop}
  \label{prop:symconc} Suppose $\sconfig\sterm\sweight U$ is a symbolic configuration, and $\args\in U$. Let $M\equiv\conc\sterm\args$ and $w\defeq\sweight\args$. Then
  \begin{enumerate}
    \item If $\sconfig\sterm\sweight U\sred\sconfig\stermb{\sweight'}{U'}$ and $\argsc\in U'$ then
    \[
    \config M w\trace\red\config{\conc\stermb\argsc}{\sweight\argsb}{\trace\concat\traceb}.
    \]
    \item\label{prop:corcompl} If $\config M w\trace\red\config N{w'}{\tupf{s'}}$ then there exists $\sconfig\sterm\sweight U\sred\sconfig\stermb{\sweight'}{U'}$ such that $\conc\stermb\argsb\equiv N$, $\sweight'\argsb=w'$ and $\argsb\in U'$. 
  \end{enumerate}
\end{restatable}

As a consequence of \cref{lem:instval}, we obtain a proof of \cref{thm:soundness and completeness}.

\section{Densities of Almost Surely Terminating Programs are Differentiable Almost Everywhere}
\label{subsec:main result}



So far we have seen that the symbolic execution semantics provides a sound and complete way to reason about the weight and value functions.
In this section we impose further
restrictions on the primitive operations and the terms to obtain results about the differentiability of these functions.

Henceforth we assume \cref{ass:pop} and we fix a term $M$ with free variables amongst $x_1,\ldots,x_m$.

\cm{A sentence explaining this lemma?}
From \cref{lem:seva} we immediately obtain the following:
\begin{restatable}{lemma}{lemmavalid}
  \label{lem:valid}
  Let $\sconfig\sterm\sweight U$ be a symbolic configuration such that $\sweight$ is differentiable on $\interior U$ and $\mu(\boundary U)=0$.
  If $\sconfig\sterm\sweight U\sred\sconfig{\sterm'}{\sweight'}{U'}$ then $\sweight'$ is differentiable on $\interior{U'}$ and $\mu(\boundary{U'})=0$.
\end{restatable}

  \subsection{Differentiability on Terminating Traces}

As an immediate consequence of the preceding, \cref{lem:seva} and the Soundness (\cref{it:sound} of \cref{thm:soundness and completeness}), whenever $\sconfig M\const{\Real^m}\sred^*\sconfig\sval\sweight U$ then $\weightfn_M$ and $\valuefn_M$ are differentiable everywhere in $\interior U$.

Recall the 
set $\trtermM$ of $\args\in\Real^m\times\traces$ from \cref{eq:trtermM} for which $M$ terminates. We abbreviate $\trtermM$ to $\trterm$ and define
\begin{align*}
    \trterm&\defeq \trtermM = \{\args \in \Real^m \times \traces \mid \exists V, w \,.\, \config {M[{\tupf r} / {\tupf x}]}1\emptytrace\red^*\config V w\trace\}\\
    \trtermi&\defeq \bigcup\{\interior{\trb} \mid \exists \sval, \sweight \, . \,\sconfig M\const{\Real^m}\sred^*\sconfig\sval\sweight{\trb}\}
  \end{align*}
  By Completeness (\cref{it:compl} of \cref{thm:soundness and completeness}),
  $\trterm=\bigcup\{\trb\mid \exists \sval, \sweight \, . \, \sconfig M\const{\Real^m}\sred^*\sconfig\sval\sweight{\trb}\}$. Therefore, being countable unions of measurable sets (\cref{lem:basic,lem:pres}), $\trterm$ and $\trtermi$ are measurable.

By what we have said above, $\weightfn_M$ and $\valuefn_M$ are differentiable everywhere on $\trtermi$.  Observe that in general, $\trtermi\subsetneq\interior\trterm$. However,
\begin{align}
  \label{eq:termi}
  \mu\left(\trterm\setminus\trtermi\right) &=\mu\bigg(\displaystyle \bigcup_{{\substack{U : \sconfig M\const{\Real^m}\sred^*\\\sconfig {\sval} \sweight U}}} \big(U \setminus \interior U \big)\bigg)\leq\displaystyle \sum_{{\substack{U : \sconfig M\const{\Real^m}\sred^*\\\sconfig {\sval} \sweight U}}} \;\mu(\boundary U)=0
\end{align}
The first equation holds because the $U$-indexed union is of pairwise disjoint sets. The inequality is due to $(U \setminus \interior{U}) \subseteq \boundary{U}$.
The last equation above holds because each $\mu (\boundary U) = 0$ (\cref{ass:pop} and \cref{lem:valid}).

Thus we conclude:
\begin{theorem}
  \label{thm:diffterm}
  Let $M$ be an SPCF term. Then its weight function $\weightfn_M$ and value function $\valuefn_M$ are differentiable for almost all terminating traces.
\end{theorem}

\subsection{Differentiability for Almost Surely Terminating Terms}
\label{sec:diffast}
Next, we would like to extend this insight for almost surely terminating terms to suitable subsets of $\Real^m\times\traces$, the union of which constitutes almost the entirety of $\Real^m\times\traces$.
Therefore, it is worth examining consequences of almost sure termination (see \cref{def:ast}).

We say that $(\tupf r,\tupf s)\in\Real^m\times\traces$ is \defn{maximal} (for $M$) if $\config {M[\tupf{\PCF r}/\tupf x]} 1\emptytrace\red^*\config N w{\tupf s}$ and
for all
\changed[cm]{
$\tupf{s'} \in \traces\setminus \{\emptytrace\}$}
and $N'$,
$\config N w{\tupf{s}}\not\red^*\config{N'}{w'}{\tupf{s}\concat\tupf {s'}}$.
Intuitively, $\tupf s$ contains a maximal number of samples to reduce $M[\tupf r/\tupf x]$.
\changed[dw]{Let $\trmax$ be the set of maximal $\args$}.

Note that $\trterm\subseteq\trmax$ and there are terms for which the inclusion is strict (e.g. for the diverging term $M\equiv \Y(\lambda f\ldotp f)$, $\emptytrace\in\trmax$ but $\emptytrace\not\in\trterm$).
\dw{I changed this back as $\trmax$ is not explicitly used in the conclusion any more.}
\lo{OK}
Besides, $\trmax$ is measurable because, thanks to \cref{prop:symconc}, for every $n\in\Nat$,
\begin{align*}
  \left\{\args\in\Real^m\times\traces_n \mid \config {M[\tupf{\PCF r}/\tupf x]} 1\emptytrace\red^*\config N w{\tupf s}\right\}=\bigcup_{{\substack{U : \sconfig M\const{\Real^m}\sred^*\\\sconfig {\stermb} \sweight U}}} U\cap(\Real^m\times\traces_n)
\end{align*}
and the RHS is a countable union of measurable sets (\cref{lem:basic,lem:pres}).

The following is a consequnce of the definition of almost sure termination and a corollary of Fubini's theorem (see \ifproceedings\cite{MOPW20} \else\cref{app:diffast} \fi for details):
\begin{restatable}{lemma}{maxtermlemma}
  \label{lem:maxterm}
  If $M$ terminates almost surely then $\mu(\trmax\setminus\trterm)=0$.
\end{restatable}


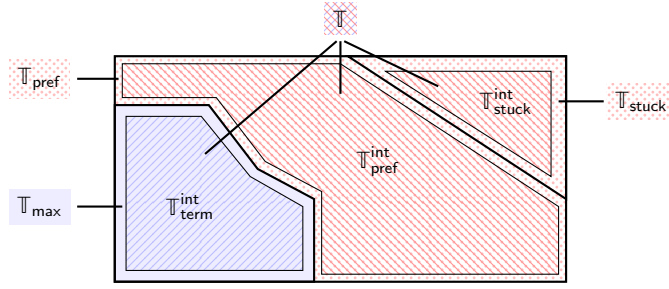
\begin{figure}[t]
  \centering
  \begin{tikzpicture}
    \draw[pattern color=red!25, pattern=crosshatch dots] (0,3) -- (6,3) -- (6,0) -- (2.65,0) -- (2.65,1.1) -- (1.9,1.5) -- (1.25,2.35) -- (0,2.35) -- cycle;

    \draw[thick,fill=blue!7] (0,0) -- (2.65,0) -- (2.65,1.1) -- (1.9,1.5) -- (1.25,2.35) -- (0,2.35) -- cycle;
    \node [align=center,fill=blue!7] at (-1,1){$\trmax$};
    \draw[thick] (-0.5,1) -- (0.1,1);

    \draw[pattern=north east lines,pattern color=blue!25] (0.15,0.15) -- (2.5,0.15) -- (2.5,1) -- (1.8,1.4) -- (1.15,2.2) -- (0.15,2.2) -- cycle;
    \node [align=center] at (1,1){$\trtermi$};

    \draw[pattern color=red!40, pattern=north west lines] (5.8,1.4) -- (5.8,2.8) -- (3.6,2.8) -- cycle;
    \node [align=center] at (5.2,2.4){$\trexti$};

    \draw[thick] (3.1,3) -- (6,1.1) -- (6,3) -- cycle;
    \node [align=center,pattern color=red!25, pattern=crosshatch dots] at (7,2.4){$\trext$};
    \draw[thick] (6.5,2.4) -- (5.9,2.4);

    \draw[pattern color=red!40, pattern=north west lines] (0.1,2.9) -- (3,2.9) -- (5.9,1) -- (5.9,0.1) -- (2.75,0.1) -- (2.75,1.2) -- (2,1.6) -- (1.35,2.45) -- (0.1,2.45) -- cycle;
    \node [align=center] at (3.5,1.6){$\trprefi$};

    \node [align=center,pattern color=red!25, pattern=crosshatch dots] at (-1,2.7){$\trpref$};
    \draw[thick] (-0.5,2.7) -- (0.05,2.7);

    \node [align=center,preaction={pattern=north east lines,pattern color=blue!25},pattern color=red!40, pattern=north west lines] at (3,3.5){$\tr$};
    \draw[thick] (3,3.2) -- (3,2.5);
    \draw[thick] (2.95,3.2) -- (1.2,1.7);
    \draw[thick] (3.05,3.2) -- (4.3,2.6);

    \draw[thick] (0,0) -- (6,0) -- (6,3) -- (0,3) -- cycle;
  \end{tikzpicture}

  \caption{\label{fig:newintm0} Illustration of how $\Real^m\times\traces$ -- visualised as the entire rectangle -- is partitioned to prove \cref{thm:mainres}.
  The value function returns $\bot$ in the red dotted area and a closed value elsewhere (i.e.\ in the blue shaded area).}
\end{figure}

Now, observe that for all $(\tupf r,\tupf s)\in\Real^m\times\traces$, exactly one of the following holds:
\begin{enumerate}[noitemsep,nolistsep]
\item $\args$ is maximal
\item for a proper prefix $\tupf{s'}$ of $\tupf s$, $\argsb$ is maximal
\item $\args$ is \emph{stuck}, because $\tupf s$ does not contain enough randomness\dw{N.B.\ undefined function applications and negative scoring is already dealt with as maximal traces.}.
\end{enumerate}
Formally, we say $\args$ is \defn{stuck}  if $\config {M[\tupf{\PCF r}/\tupf x]} 1\emptytrace\red^*\config{E[\Sample]}w{\tupf s}$, and we let $\trext$ be the set of all $\args$ which get stuck. Thus,
\begin{align*}
  \Real^m\times\traces=\trmax\cup\trpref\cup\trext
\end{align*}
where $\trpref\defeq\{\argsc\mid\args\in\trmax\land\traceb\neq\emptytrace\}$,
and the union is disjoint.

 Defining $\trexti\defeq \bigcup\{\interior{\trb}\mid \sconfig M\const{\Real^m}\sred^*\sconfig {\scon[\Sample]}\sweight\trb\}$
we can argue analogously to \cref{eq:termi} that $\mu(\trext\setminus\trexti)=0$.

Moreover, for $\trprefi\defeq\{(\tupf r,\tupf s\concat\tupf{s'})\mid(\tupf r,\tupf s)\in\trtermi\tand\emptytrace\neq\tupf{s'}\in\traces\}$ it holds
\begin{align*}
  \trpref\setminus\trprefi=\bigcup_{n\in\Nat}\left\{(\tupf r,\tupf s\concat\tupf{s'})\mid (\tupf r,\tupf s)\in\trmax\setminus\trtermi\land\tupf{s'}\in\traces_n\right\}
\end{align*}
and hence, $\mu(\trpref\setminus\trprefi)\leq\sum_{n\in\Nat}\mu(\trmax\setminus\trtermi)\leq 0$.

Finally, we define
\begin{align*}
  \tr\defeq\trtermi\cup\trprefi\cup\trexti
\end{align*}
Clearly, this is an open set and the situation is illustrated in \cref{fig:newintm0}. By what we have seen,
\begin{align*}
  \mu\left((\Real^m\times\traces)\setminus\tr\right)=\mu(\trterm\setminus\trtermi)+\mu(\trprefi\setminus\trpref)+\mu(\trext\setminus\trexti)=0
\end{align*}
Moreover, to conclude the proof of our main result \cref{thm:mainres} it suffices to note: 
\begin{enumerate}[noitemsep,nolistsep]
\item $\weightfn_M$ and $\valuefn_M$ are differentiable everywhere on $\trtermi$ (as for \cref{thm:diffterm}), and
\item $\weightfn_M(\tupf r,\tupf s)=0$ and $\valuefn_M(\tupf r,\tupf s)=\bot$ for $(\tupf r,\tupf s)\in\trprefi\cup\trexti$.
\end{enumerate}

\mainres

\fi

We remark that almost sure termination was not used in our development until the proof of \cref{lem:maxterm}.
For \cref{thm:mainres} we could have instead directly assumed the conclusion of \cref{lem:maxterm}; that is, almost all maximal traces are terminating.
 This is a strictly weaker condition than almost sure termination. The exposition we give is more appropriate: almost sure termination is a standard notion, and the development of methods to prove almost sure termination is a subject of active research.

We also note that the technique used in this paper to  establish almost everywhere differentiability  could be used to target another ``almost everywhere'' property instead: one can simply remove the requirement that elements of $\pop$ are differentiable, and replace it with the desired property. A basic example of this is \emph{smoothness}.

\section{Conclusion}

\label{sec:conclusion}


We have solved an open problem in the theory of probabilistic programming.
This is mathematically interesting, and motivated the development of stochastic symbolic execution, a more informative form of operational semantics in this context.
The result is also of major practical interest, since almost everywhere differentiability is necessary for correct gradient-based inference.

\subsubsection*{Related Work.}
This problem was partially addressed in the work of Zhou et al.~\cite{DBLP:conf/aistats/ZhouGKRYW19} who prove a restricted form of our theorem for recursion-free first-order programs with analytic primitives.
Our stochastic symbolic execution is related to their \emph{compilation scheme}, which we extend to a more general language.

The idea of considering the possible control paths through a probabilistic programs is fairly natural and not new to this paper; it has been used towards the design of specialised inference algorithms for probabilistic programming, see \cite{chaganty2013efficiently,DBLP:journals/corr/abs-1910-13324}. To our knowledge, this is the first semantic formalisation of the concept, and the first time it is used to reason about whole-program density.

The notions of \emph{weight function} and \emph{value function} in this paper are inspired by the more standard trace-based operational semantics of
Borgstr{\"{o}}m et al.~\cite{DBLP:conf/icfp/BorgstromLGS16} (see also \cite{DBLP:journals/pacmpl/WandCGC18,DBLP:journals/pacmpl/LeeYRY20}).


Mazza and Pagani \cite{MP21} study the correctness of automatic differentiation (AD) of purely \emph{deterministic} programs.
\changed[lo]{This problem is orthogonal to the work reported here, but it is interesting to combine their result with ours.}
Specifically, we show a.e.~differentiability whilst \cite{MP21} proves a.s.~correctness of AD on the \emph{differentiable} domain. Combining both results one concludes that for a deterministic program, AD returns a correct gradient a.s.~on the \emph{entire} domain. \changed[hp]{Going deeper into the comparison, Mazza and Pagani propose a notion of admissible primitive function strikingly similar to ours: given continuity, their condition 2 and our condition \ref{ass:boundary} are equivalent. On the other hand we require admissible functions to be differentiable, when they are merely continuous in \cite{MP21}. Finally, we conjecture that ``stable points'', a central notion in \cite{MP21}, have a clear counterpart within our framework: for a symbolic evaluation path arriving at $\sconfig \sval w U$, for $\sval$ a symbolic value, the points of $\interior U$ are precisely the stable points. }

Our work is also connected to recent developments in differentiable programming. Lee et al.~\cite{DBLP:journals/corr/abs-2006-06903} study the family of \emph{piecewise functions under analytic partition}, or just ``PAP'' functions. PAP functions are a well-behaved family of almost everywhere differentiable functions, which can be used to reason about automatic differentiation in recursion-free first-order programs. An interesting question is whether this can be extended to a more general language, and whether densities of almost surely terminating SPCF programs are PAP functions. (See also \cite{DBLP:conf/fossacs/HuotSV20,DBLP:journals/pacmpl/BrunelMP20} for work on differentiable programs \emph{without} conditionals.)

A similar class of functions is also introduced by Bolte and Pauwels \cite{DBLP:journals/corr/abs-2006-02080} in very recent work; this is used to prove a convergence result for stochastic gradient descent in deep learning. Whether this class of functions can be used to reason about probabilistic program densities remains to be explored.

Finally we note that \emph{open logical relations} \cite{openlogical} are a convenient proof technique for establishing properties of programs which hold at first order, such as almost everywhere differentiability. This approach remains to be investigated in this context, as the connection with probabilistic densities is not immediate.

\subsubsection*{Further Directions.}
This investigation would benefit from a denotational treatment; this is not currently possible as existing models of probabilistic programming do not account for differentiability.

In another direction, it is likely that we can generalise the main result by extending SPCF with recursive types, as in~\cite{DBLP:journals/pacmpl/VakarKS19}, and, more speculatively, first-class differential operators as in~\cite{DBLP:journals/tcs/EhrhardR03}. It would also be useful to add to SPCF a family of \emph{discrete} distributions, and more generally continuous-discrete mixtures, which have practical applications \cite{narayanan2020symbolic}.

Our work will have interesting implications in the correctness of various gradient-based inference algorithms, such as the recent discontinuous HMC \cite{10.1093/biomet/asz083} and reparameterisation gradient for non-differentiable models \cite{DBLP:conf/nips/0001YY18}.
But given the lack of guarantees of correctness properties available until now, these algorithms have not yet been developed in full generality, leaving many perspectives open.

\paragraph{Acknowledgements.}
\changed[lo]{We thank Wonyeol Lee for spotting an error in an example.}

We gratefully acknowledge support from EPSRC and the Royal Society. 

\bibliography{bib}


\vfill

{\small\medskip\noindent{\bf Open Access} This chapter is licensed under the terms of the Creative Commons\break Attribution 4.0 International License (\url{http://creativecommons.org/licenses/by/4.0/}), which permits use, sharing, adaptation, distribution and reproduction in any medium or format, as long as you give appropriate credit to the original author(s) and the source, provide a link to the Creative Commons license and indicate if changes were made.}

{\small \spaceskip .28em plus .1em minus .1em The images or other third party material in this chapter are included in the chapter's Creative Commons license, unless indicated otherwise in a credit line to the material.~If material is not included in the chapter's Creative Commons license and your intended\break use is not permitted by statutory regulation or exceeds the permitted use, you will need to obtain permission directly from the copyright holder.}

\medskip\noindent\includegraphics{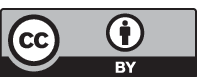}

\ifproceedings
\else
\appendix


\section{Supplementary Materials}

\subsection{More Details on \Cref{eg:primitive-functions}}
\label{apx:pop}
We prove that $\pop_1$ and $\pop_2$ of \Cref{eg:primitive-functions} satisfy \Cref{ass:pop}.

  \begin{enumerate}
  \item

      Clearly, all constant and projection functions are analytic.
      Since analytic functions are total and differentiable (hence continuous) functions, they are \changed[lo]{Borel-measurable}.
      \lo{Reason: continuous functions are of course not measurable in general.}
      Therefore, due to the fact that analytic functions are closed under pairs and composition \cite[Prop.~2.4]{Conway},
      it remains to check whether the boundary of $f^{-1}([0,\infty))$ has measure zero.


    Since $\inv{f}(\interior{A}) \subseteq \interior{\inv{f}(A)}$ and
    $\closure{\inv{f}(A)} \subseteq \inv{f}{(\closure{A})}$
    for any subset $A \subseteq \Real$,
    we have
    $\boundary f^{-1}(A)
      \subseteq \inv{f}{(\closure{A})}\setminus \inv{f}{(\interior{A})}
      = \inv{f}(\boundary{A}) $.
    Letting $A = [0,\infty)$, we have
    $\boundary \inv{f}({[0,\infty)}) \subseteq \inv{f}(\{0\})$.
    Applying the well-known result \cite{mityagin2015zero} that
    the zero set of all analytic functions, except the zero function, has measure zero
    we conclude that $\boundary \inv{f}({[0,\infty)})$ has measure zero.
    It is easy to see that
    if $f$ is the zero function,
    $\boundary \inv{f}([0,\infty)) = \boundary \Real^n = \varnothing$ has measure zero.


  \item Clearly all constant functions and projections are in $\pop_2$.

    Note that the set of finite unions of (possibly unbounded) rectangles forms an algebra $\mathcal A$ (i.e.~a collection of subsets of $\Real^n$ closed under complements and finite unions, hence finite intersections).
  Then $\dom f=f^{-1}(-\infty,0]\cup f^{-1}(0,\infty)\in\mathcal A$.
  Besides, for every $U\in\mathcal A$, $\leb(\boundary U)=0$ (because $\leb(\boundary R) = 0$ for every rectangle $R$).

 \changed[lo]{It remains to prove that $\pop_2$ is closed under composition.
 Suppose that $f\from\Real^\ell\pto\Real\in\pop_2$ and $g_1,\ldots,g_\ell\from\Real^m\pto\Real\in\pop_2$.
 It is straightforward to see that $\dom{f \circ \langle g_i \rangle_{i = 1}^\ell}$ remains open (note that $\circ$ is relational composition).
 Moreover, for every $x$ in $\dom{f \circ \langle g_i \rangle_{i = 1}^\ell}$, $\langle g_i \rangle_{i = 1}^\ell (x)$ is an interior point in $\dom f$ which is open.
 It follows that the (standard) chain rule is applicable, and we have $f \circ \langle g_i \rangle_{i = 1}^\ell$ is differentiable at $x$.}
 Besides, suppose $I$ is a (possibly unbounded) interval. By assumption there are $m\in\Nat$ and (potentially unbounded) intervals $I_{i,j}$, where $1\leq i\leq m$ and $1\leq j\leq\ell$ such that $f^{-1}(I)=\bigcup_{i=1}^m I_{i,1}\times\cdots\times I_{i,\ell}$. Observe that
    \begin{align*}
      & \left(f\circ \langle g_i \rangle_{i = 1}^\ell \right)^{-1}(I)\\
      &=\left\{\tupf r\in\dom{g_1}\cap\cdots\cap\dom{g_\ell}\mid (g_1(\tupf r),\ldots,g_\ell(\tupf r))\in f^{-1}(I)\right\}\\
      &=\bigcup_{i=1}^m\left\{\tupf r\in\dom{g_1}\cap\cdots\cap\dom{g_\ell}\mid g_1(\tupf r)\in I_{i,1}\land\cdots\land g_\ell(\tupf r)\in I_{i,\ell}\right\}\\
      &=\bigcup_{i=1}^m g_1^{-1}(I_{i,1})\cap\cdots\cap g^{-1}_\ell(I_{i,\ell})
    \end{align*}
    and this is in $\mathcal A$ because algebras are closed under finite unions and intersections.
  \end{enumerate}

\subsection{Supplementary Materials for \Cref{sec:stermval}}
\label{app:stermval}
\instval*
\begin{proof}[Proof sketch]
  First, suppose $\sterm$ is a symbolic value $\sval$. It is easy to prove inductively that for a sybolic values $\sval$ of type $\PCFReal$, $\conc\sval(\tupf r,\tupf s)$ is a real constant. Otherwise both $\sval$ and $\conc\sval(\tupf r,\tupf s)$ are abstractions.

  Conversely, suppose $\conc\sval(\tupf r,\tupf s)$ is a value. If it is an abstraction then so is $\sval$. Otherwise it is a real constant. By the definition of $\conc\cdot$ and a case inspection of $\sterm$ this is only possible if $\sterm$ is a symbolic value.
\end{proof}

The following substitution property holds for symbolic terms $\sterm$ and $\stermb$:
\begin{align}
  \label{eq:subst}
  \domnp\conc{\sterm[\stermb/y]}&\subseteq\domnp\conc\sterm\cap\domnp\conc\stermb\\
  \conc\sterm\args[\conc\stermb\args/y]&\equiv\conc{\sterm[\stermb/y]}(\tupf r,\tupf s)
\end{align}

\technicalprimitives*
\begin{proof}
  \begin{enumerate}
  \item 
  We prove the first part by induction on symbolic values.
    \begin{itemize}
    \item For $r'\in\Real$, $\seva{\PCF {r'}}$ is a constant function and $\seva{x_i}$ and $\seva{\alpha_j}$ are projections, which are in $\pop$ by assumption.
    \item Next, suppose $\sval$ is a symbolic value $\delay f(\sval_1,\ldots,\sval_\ell)$. By the inductive hypothesis, each $\seva{\sval_i}\in\pop$.
      It suffices to note that $\seva{\delay f(\sval_1,\ldots,\sval_\ell)}(\tupf r,\tupf s)=f(\seva{\sval_1}(\tupf r,\tupf s),\ldots,\seva{\sval_\ell}(\tupf r,\tupf s))$ for $(\tupf r,\tupf s)\in\domnp \delay f(\sval_1,\ldots,\sval_\ell)$. Therefore, $\seva{\delay f(\sval_1,\ldots,\sval_\ell)}$ because $\pop$ is assumed to be closed under composition.
    \item Finally, note that we do not need to consider abstractions because they do not have type $\PCFReal$.
    \end{itemize}
  \item
    \begin{itemize}
    \item     Note that $\conc{x_i}$ and $\conc{\alpha_j}$ are projection functions and $\conc{\PCF r}$ are constant functions, which are (everywhere) differentiable functions.
    \item Besides, $\delay f(\sval_1,\ldots,\sval_\ell)$ is a symbolic value and on its domain,
    \begin{align*}
      \conc{\delay f(\sval_1,\ldots,\sval_\ell)}=\bblambda (\tupf r,\tupf s)\ldotp\PCF{\seva{\delay f(\sval_1,\ldots,\sval_\ell)}(\tupf r,\tupf s)}
    \end{align*}
    $\seva{\delay f(\sval_1,\ldots,\sval_\ell)}$ is in $\pop$ by the first part and by assumption this implies differentiability.
    \item The function $\conc {\lambda y. \sterm}$ is obtained by composing $\conc \sterm$ with the function $\terms \to \terms : \termc \mapsto \lambda y. \termc$. The latter is easily seen to be differentiable: recall that $\terms =\bigcup_{n \in \mathbb N}
      \bigcup_{M} \{M\} \times \Real^n$, where $M$ ranges over skeleton terms with $n$ place-holders. On each component $\{M\} \times \Real^n$ the function acts as $(M, \tupf x) \mapsto (\lambda y. M, \tupf x)$; it is simply one of the coproduct injections, hence differentiable.
    \item  The cases of $\conc {\Y \sterm}$ and $\conc{\Score \sterm}$ are analogous.
    \item  The function $\conc {M N}$ is obtained by composing $\conc M \times \conc N$ with the diagonal map $\args \mapsto (\args, \args)$; both are differentiable.
    \item The cases of $\conc{\PCF f (\sterm_1, \dots, \sterm_\ell)}$ and $\conc {\Ifleq\stermc\sterm\stermb}$ are similar, using diagonal maps of different arities.
    \item  The function $\conc\Sample$ is a constant function, so it is differentiable. This covers all cases. \qedhere
    \end{itemize}
  \end{enumerate}
\end{proof}

\subsection{Supplementary Materials for \Cref{sec:sred}}
\label{app:sred}

\ssubcon*
\begin{proof}
  We prove all parts of the lemma simultaneously by structural induction on $\sterm$.
  \begin{itemize}[noitemsep]
  \item First, note that for every $\scon$ and $\sredex$ all of the following holds
\[
      x_i\not\equiv\scon[\sredex]\qquad\alpha_j\not\equiv\scon[\sredex]\qquad y\not\equiv\scon[\sredex]\qquad\PCF r\not\equiv\scon[\sredex]\qquad\lambda y\ldotp\sterm\not\equiv\scon[\sredex]
\]
    and the left hand sides are symbolic values.
  \item Note that $\domnp\conc{\delay f(\sterm_1,\ldots,\sterm_\ell)}=\emptyset$ unless it is a symbolic value. Besides, for every $\scon$ and $\sredex$, $\delay f(\sterm_1,\ldots,\sterm_\ell)\not\equiv\scon[\sredex]$.
  \item If $\sterm\equiv\stermb_1\,\stermb_2$ then $\sterm$ is not a symbolic value.

    Suppose that $\stermb_1$ is an abstraction. If $\stermb_2$ is a symbolic value then $\sterm$ is a symbolic redex and by the first part of the inductive hypothesis, $\sterm\equiv\scon[\sredex]$ implies $\scon\equiv[]$ and $\sredex\equiv\sterm$.

    If $\sterm_2$ is not a symbolic value then $\sterm$ is not a symbolic redex.
    Note that $\sterm\equiv\scon[\sredex]$ implies $\scon\equiv\stermb_1\,\scon'$.
    By the second part of the inductive hypothesis $\scon'$ and $\sredex$ are unique if they exist.
    Besides, due to $\domnp\conc\sterm\subseteq\domnp\conc{\stermb_2}$ and the third part of the inductive hypothesis, such $\scon'$ and $\sredex$ exist if $\domnp\conc\sterm\neq\emptyset$.

    If $\stermb_1$ is not an abstraction it cannot be a symbolic value and $\sterm\equiv\scon[\sredex]$ implies $\scon\equiv\scon'\,\stermb_2$.
    By the second part of the inductive hypothesis, $\scon'$ and $\sredex$ are unique if they exist.
    Besides, because of $\domnp\conc\sterm\subseteq\domnp\conc{\sterm_1}$ and the third part of the inductive hypothesis, such $\scon'$ and $\sredex$ exists if $\domnp\conc\sterm\neq\emptyset$.
  \item Next, suppose $\sterm\equiv\PCF f(\stermb_1,\ldots,\stermb_\ell)$, which is clearly not a symbolic value.

    If all $\stermb_i$ are symbolic values, $\sterm$ is a symbolic redex and by the first part of the inductive hypothesis, $\scon\equiv[]$ and $\sredex\equiv\sterm$ are unique such that $\sterm\equiv\scon[\sredex]$.

    Otherwise, suppose $i$ is minimal such that $\stermb_i$ is not a symbolic value. Clearly, $\sterm\equiv\scon[\sredex]$ implies $\scon\equiv\PCF f(\stermb_1,\ldots,\stermb_{i-1},\scon',\stermb_{i+1},\ldots,\stermb_\ell)$ and $\stermb_i\equiv\scon'[\sredex]$. By the second part of the inductive hypothesis $\scon'$ and $\sredex$ are unique if they exist.
    Besides due to $\domnp\conc\sterm\subseteq\domnp\conc{\stermb_i}$ and the third part of the inductive hypothesis such $\scon'$ and $\sredex$ exist if $\domnp\conc\sterm\neq\emptyset$.
  \item If $\sterm\equiv\Y\stermb$, $\sterm\equiv\Sample$ or $\sterm\equiv\Score\stermb$, which are not symbolic values, then this is obvious (using the inductive hypothesis).
  \item Finally, suppose $\sterm\equiv(\Ifleq\stermc{\stermb_1}{\stermb_2})\equiv\scon_1[\sredex_1]\equiv\scon_2[\sredex_2]$.
    If $\stermc$ is a symbolic value then $\sterm$ is a symbolic redex and by the first part of the inductive hypothesis, $\sterm\equiv\scon[\sredex]$ implies $\scon\equiv[]$ and $\sredex\equiv\sterm$.

    If $\stermc$ is not a symbolic value then $\sterm\equiv\scon[\sredex]$ implies $\scon\equiv\Ifleq{\scon'}{\stermb_1}{\stermb_2}$ and $\stermc\equiv\scon'[\sredex]$.
    By the second part of the inductive $\scon'$ and $\sredex$ are unique if they exist. Due to $\domnp\conc\sterm\subseteq\domnp\conc\stermc$ and the third part of the inductive hypothesis such $\scon'$ and $\sredex$ exist provided that $\domnp\conc\sterm\neq\emptyset$.
  \end{itemize}
\end{proof}

We obtain that for all $\scon$, $\sterm$ and $\args\in\domnp\conc{\scon[\sterm]}$:
\begin{align}
  \label{eq:substcon}
  \conc\scon\args\left[\conc\sterm\args\right]\equiv\conc{\scon[\sterm]}\args
\end{align}

\reductionvalidconfigs*
\begin{proof}
  Suppose that $\sterm\equiv\scon[\sredex]$ and $\stermb\equiv\scon[\scontra]$.
  Because of \cref{lem:seva} and the assumption that the functions in $\pop$ are measurable, $U'$ is measurable again. Furthermore, the rules ensure that $U'\subseteq\domnp\conc{\scontra}$. (For the first rule this is because of the Substitution \cref{eq:subst}.)
  By the Substitution \cref{eq:substcon}, $U'\subseteq\domnp\conc{\scon[\sredex]}\cap\domnp\conc\scontra\subseteq\domnp\conc\scon\cap\domnp\conc\scontra=\domnp\conc{\scon[\scontra]}$.
\end{proof}

\basicproperties*
\begin{proof}[Proof sketch]
  By subject construction (\cref{lem:basic0}), there is at most one $\scon$ and $\sredex$ such that $\sterm\equiv\scon[\sredex]$.
  An inspection of the rules shows that $U'$ such that $\sconfig\sredex\sweight U\sred\sconfig\scontra{\sweight'}{U'}$ is unique unless $\sredex$ is a conditional, in which case there are two distinct such $U'$. Hence there are at most two distinct $U'$ such that $\sconfig{\scon[\sredex]}\sweight U\sred\sconfig\stermb{\sweight'}{U'}$. The first part follows by induction on the number of reduction steps.

  For the other two parts note that if $\sconfig\sterm\sweight U\sred\sconfig\stermb{\sweight'}{U'}$ either $U'\subseteq U$ or $U'=\{(\tupf r,\tupf s\concat[r'])\mid\args\in U\land r'\in(0,1)\}$.
  In particular, if $\sconfig\sterm\sweight U\sred^*\sconfig\stermb{\sweight'}{U'}$, $U'\subseteq\{(\tupf r,\trace\concat\traceb)\mid\args\in U\land\traceb\in\traces_n\}$ for some $n\in\Nat$.


  By the discussion for the first part of the lemma, if $\config\sterm\sweight U\sred^*\sconfig{\stermb_i}{\sweight_i}{U_i}$ for $i\in\{1,2\}$, then \emph{w.l.o.g.},  either
  \begin{enumerate}
  \item $\sconfig{\stermb_1}{\sweight_1}{U_1}\sred^*\sconfig{\stermb_2}{\sweight_2}{U_2}$ or
  \item $\sconfig\sterm\sweight U\sred^*\sconfig{\scon[\Ifleq\stermc{\sterm_1}{\sterm_2}]}{\sweight'}{U'}$ and
    \[
      \begin{tikzcd}[row sep=-0.8em,column sep=1.2em]
        & \sconfig{\scon[\sterm_1]}{\sweight'}{U'\cap\seva\stermc^{-1}(-\infty,0]} \\[.2em]
        & \sred^*\sconfig{\stermb_1}{\sweight_1}{U_1}
        \\
        \sconfig{\scon[\Ifleq\stermc{\sterm_1}{\sterm_2}]}{\sweight'}{U'}
        \arrow[phantom]{uur}{{\rotatebox[origin=c]{+25}{$\sred^*$}}}
        \arrow[phantom]{dr}{\rotatebox[origin=c]{-25}{$\sred^*$}}
        &  \\
        & \sconfig{\scon[\sterm_2]}{\sweight'}{U'\cap\seva\stermc^{-1}(0,\infty)} \\[.2em]
        & \sred^*\sconfig{\stermb_2}{\sweight_2}{U_2}
      \end{tikzcd}
    \]
    for suitable $\stermb$, $\scon$, $\stermc$, $\sterm_1$, $\sterm_2$, $\sweight'$ and $U'$.
  \end{enumerate}
  In the latter case in particular $U_1\cap U_2=\emptyset$ holds.

  This implies the second and third part of the lemma.
\end{proof}

\begin{lemma}
  \label{lem:symcon1aux}
  Suppose $\sconfig\sredex\sweight U\sred\sconfig{\scontra}{\sweight'} {U'}$, $\args\in U$ and $\argsc\in U'$. Then $\config{\conc\sredex\args}{\sweight\args}\trace\red\config{\conc\scontra\argsc}{\sweight'\argsc}{\trace\concat\traceb}$.
\end{lemma}
\begin{proof}
  We prove the lemma by case analysis on the symbolic redex contractions.
  \begin{itemize}
    \item
      First, suppose
      $\sconfig\Sample\sweight U\sred\sconfig{\alpha_{n+1}}{\sweight'}{U'}$. Note that $s'=[r']$ for some $0<r'<1$.
      Then
      $\config\Sample{\sweight(\tupf r,\tupf s)}{\tupf s}\red
      \config{\PCF{r'}}{\sweight\args}{\tupf s\concat[r']}$ and $\sweight\args=\sweight'\argsc$ and
      $\conc{\alpha_{n+1}}\argsc\equiv \PCF{r'}$.

    \item
      Suppose $\sconfig{\Score\sval}\sweight U\sred
      \sconfig\sval{\sweight\cdot\seva\sval}{U\cap\seva\sval^{-1}[0,\infty)}$.
      Then $\args=\argsc\in U\cap\seva\sval^{-1}[0,\infty)$.
      Hence, there must exist $r'\geq 0$ such that $\conc\sval\args\equiv\conc\sval\argsc\equiv\PCF{r'}$.
      Besides,
      $\config{\Score{\conc\sval(\tupf r,\tupf s)}}{\sweight(\tupf r,\tupf s)}{\tupf s}\red
        \config{\PCF{r'}}{\sweight\args\cdot r'}{\tupf s}
      $
      and $(\sweight\cdot\seva\sval)\argsc=\sweight\args\cdot r'$.
    \item
      Suppose $\sconfig{\Ifleq\sval\sterm\stermb}\sweight U\sred\sconfig\sterm\sweight{U\cap\seva\sval^{-1}(-\infty,0]}$. Note that $\args=\argsc\in U\cap\seva\sval^{-1}(-\infty,0]$. Thus, $\seva\sval\args\leq 0$.
      Therefore,
      \begin{align*}
        &\config{\Ifleq{\PCF{\seva\sval(\tupf r,\tupf s)}}{\conc\sterm(\tupf r,\tupf s)}{\conc\stermb(\tupf r,\tupf s)}}{\sweight(\tupf r,\tupf s)}{\tupf s}\\
        &\red\config{\conc\sterm(\tupf r,\tupf s)}{\sweight(\tupf r,\tupf s)}{\tupf s}
      \end{align*}
    \item
      Similar for the else-branch.
    \item
      Suppose
      $\sconfig{(\lambda y\ldotp\sterm)\,\sval}\sweight U\red
      \sconfig{\sterm[\sval/y]}\sweight U$. Then $\args=\argsc\in U$.
      By \cref{lem:instval}, $\conc\sval\args$ is a value. Hence,
      \begin{align*}
        &\config{(\lambda y\ldotp\conc\sterm(\tupf r,\tupf s))\,\conc\sval(\tupf r,\tupf s)}{\sweight(\tupf r,\tupf s)}{\tupf s}\\
        &\red\config{(\conc\sterm(\tupf r,\tupf s))[\conc\sval(\tupf r,\tupf s)/y]}{\sweight(\tupf r,\tupf s)}{\tupf s}.
      \end{align*}
      Besides, by \cref{eq:subst},
      $(\conc\sterm(\tupf r,\tupf s))[\conc\sval(\tupf r,\tupf s)/y]\equiv
      \conc{\sterm[\sval/y]}\argsc$.
    \item Suppose $\sconfig{\PCF f(\sval_1,\ldots,\sval_\ell)}\sweight U\sred
      \sconfig{\delay f(\sval_1,\ldots,\sval_\ell)}\sweight{U\cap\domnp\seva{\PCF f(\sval_1,\ldots,\sval_\ell)}}$. Then $\args=\argsc\in U\cap\domnp\seva{\PCF f(\sval_1,\ldots,\sval_\ell)}$. In particular, $\args\in\domnp\seva{\sval_i}$ for each $1\leq i\leq\ell$ and $(\seva{\sval_1}\args,\ldots\seva{\sval_\ell}\args)\in\dom f$. Therefore,
      \begin{align*}
        &\config{\PCF f(\PCF{\seva{\sval_1}(\tupf r,\tupf s)},\ldots, \PCF{\seva{\sval_\ell}\args})}{\sweight\args}\trace\\
        &\red \config{\PCF{f(\seva{\sval_1}\args,\ldots,\seva{\sval_\ell}\args)}}{\sweight\args}\trace
      \end{align*}
      because $\PCF{ f(\seva{\sval_1}(\tupf r,\tupf s),\ldots, \seva{\sval_\ell}(\tupf r,\tupf s))}\equiv\conc{\delay f(\sval_1,\ldots,\sval_\ell)}\argsc$.

    \item Finally, suppose
      $\sconfig{\Y (\lambda x\ldotp \sterm)}\sweight U\sred
      \sconfig{\lambda y\ldotp \sterm\,[\Y (\lambda x\ldotp\sterm)/x]\, y}\sweight U$.
      Then $\args=\argsc\in U$. It holds
      \begin{align*}
        &\config{\Y{(\lambda x\ldotp\conc{\sterm}(\tupf r,\tupf s))}}
          {\sweight(\tupf r,\tupf s)}
          {\tupf s} \\
        &\red
        \config{
          \lambda y\ldotp
          \conc{\sterm}(\tupf r,\tupf s)
          [\Y{(\lambda x\ldotp\conc{\sterm}(\tupf r,\tupf s))}/x]\,y
        }{\sweight(\tupf r,\tupf s)}{\tupf s}
      \end{align*}
      and by the Substitution \cref{eq:subst},
\[
        \lambda y\ldotp
        \conc{\sterm}(\tupf r,\tupf s)
        [\Y{(\lambda x\ldotp\conc{\sterm}(\tupf r,\tupf s))}/x]\,y\equiv        \conc{
          \lambda y\ldotp\sterm[\Y{(\lambda x.\sterm)}/x]\,y
        }\argsc
\]
  \end{itemize}
\end{proof}



\begin{lemma}
  \label{lem:consym1aux}
  Suppose
  $\conc\sterm(\tupf r,\tupf s)\equiv\redexa$,
  $(\tupf r,\tupf s)\in U$,
  $\sweight\args=w$ and 
  $\config\redexa w {\tupf s}\red\config\contra{w'}{\tupf {s'}}$. Then there exists $\sconfig\sterm\sweight U\sred\sconfig{\scontra}{\sweight'}{U'}$ such that 
  $\conc{\scontra}(\tupf r,\tupf{s'})\equiv\contra$,
  $\sweight'\argsb=w'$
  and $(\tupf r,\tupf{s'})\in U'$.
\end{lemma}
\begin{proof}
  We prove the lemma by a case distinction on the redex contractions.
  \begin{itemize}
    \item
      First, suppose
      $\config\Sample w{\tupf s}\red\config{\PCF {r}} w {\tupf s\concat[r]}$,
      where $0<r<1$ and
      $(\tupf r,\tupf s)\in U\subseteq\Real^m\times\traces_n$.
      Then $\sconfig\Sample\sweight U\sred\sconfig {\alpha_{n+1}}{\sweight'} {U'}$,
      where
      $U'=\{(\tupf r,\tupf s\concat[r'])\mid\args\in U\land 0<r'<1\}\ni(\tupf r,\tupf s\concat[r])$ and
      $\sweight'(\tupf r,\trace\concat[r])=\sweight\args=w$.
      By definition,
      $\conc{\alpha_{n+1}}(\tupf r,\tupf s\concat[r])\equiv\PCF{r}$.

    \item
      Suppose
      $\config{\Score{\PCF{r'}}} w {\tupf s}\red\config{\PCF{r'}}{w'}{\tupf s}$,
      where $r'\geq 0$, $(\tupf r,\tupf s)\in U$.
      Then $\sterm\equiv\Score\sval$ for some $\sval$ satisfying
      $\conc\sval(\tupf r,\tupf s)\equiv\PCF{r'}$.
      Hence,
      $\seva\sval(\tupf r,\tupf s) = r' \geq 0$ and
      $\sconfig{\Score{\sval}}\sweight U\sred\sconfig{\sval}{\sweight\cdot\seva\sval} {U'}$ where
      $U' := {U\cap\seva\sval^{-1}[0,\infty)} \ni (\tupf r,\tupf s)$ and $(\sweight\cdot\seva\sval)\args=w\times r'$.

    \item
      Suppose
      $\config{\Ifleq{\PCF{r'}} M N} w{\tupf s}\red\config {M} w{\tupf s}$
      because $r'\leq 0$ and $(\tupf r,\tupf s)\in U$.
      Suppose that $\sval,\sterm,\stermb$ are such that
      $\conc\sval(\tupf r,\tupf s)\equiv\PCF{r'}$,
      $\conc\sterm(\tupf r,\tupf s)\equiv M$ and
      $\conc\stermb(\tupf r,\tupf s)\equiv N$. Then
      $\seva\sval(\tupf r,\tupf s)=r'\leq 0$ and
      $\sconfig{\Ifleq\sval\sterm\stermb}\sweight U\sred\sconfig\sterm\sweight{U'}$, where $U'\defeq (U\cap\seva\sval^{-1}(-\infty,0])\ni(\tupf r,\tupf s)$ by assumption.

    \item
      Similar for the else-branch.

    \item
      Suppose
      $\config{\PCF f(r'_1,\ldots,r'_\ell)} w{\tupf s}\red
      \config{\PCF{f(r'_1,\ldots,r'_\ell)}} w{\tupf s}$,
      where $(r'_1,\ldots,r'_\ell)\in\dom f$.
      Suppose further that $\sval_1,\ldots,\sval_\ell$ are such that
      for each $1\leq i\leq\ell$,
      $\conc{\sval_i}(\tupf r,\tupf s)\equiv\PCF{r'_i}$.
      Since
      $\domnp\seva{\PCF f(\sval_1,\ldots,\sval_\ell)} =
      \domnp\conc{\PCF f(\sval_1,\ldots,\sval_\ell)}$, then we have
      $\sconfig{\PCF f(\sval_1,\ldots,\sval_\ell)}\sweight U
      \sred
      \sconfig{\delay f(\sval_1,\ldots,\sval_\ell)}\sweight {U'}$,
      where $U'\defeq(U\cap\domnp\seva{\PCF f(\sval_1,\ldots,\sval_\ell)})
      \ni(\tupf r,\tupf s)$.

    \item
      Suppose
      $\config{(\lambda y\ldotp M)\,V}w{\tupf s}\red
      \config{M[V/y]}w{\tupf s}$ and
      $(\tupf r,\tupf s)\in U$.
      Let $\sterm,\stermb$ be such that
      $\conc\sterm(\tupf r,\tupf s)\equiv M$ and
      $\conc\stermb(\tupf r,\tupf s)\equiv V$.
      By the Substitution \cref{eq:subst},
      $\conc{\sterm[\stermb/y]}(\tupf r,\tupf s)\equiv M[N/y]$ and
      by \cref{lem:instval}, $\stermb$ must be a symbolic value.
      Thus,
      $\sconfig{(\lambda y\ldotp\sterm)\,\stermb}\sweight U\sred
      \sconfig{\sterm[\stermb/y]}\sweight U$.

    \item
      Suppose
      $\config{\Y{(\lambda y.\terma)}} w{\tupf s}\red
      \config{\lambda z.\terma[\Y{(\lambda y.\terma)}/y]z} w{\tupf s}$ and
      $(\tupf r,\tupf s)\in U$.
      Let $\sterm$ be such that
      $\conc\sterm (\tupf r,\tupf s) \equiv \terma$,
      then by Substitution \cref{eq:subst}, we have
      $\conc{
        \lambda z.\sterm[\Y{(\lambda y.\sterm)}/y]z
      }(\tupf r,\tupf s) \equiv
      \lambda z.\terma[\Y{(\lambda y.\terma)}/y]z
      $. Thus we conclude that
      $\sconfig{\Y{(\lambda y.\sterm)}}\sweight U\sred
      \sconfig{\lambda z.\sterm[\Y{(\lambda y.\sterm)}/y]z}\sweight U$.
  \end{itemize}
\end{proof}

\begin{definition}
  \label{def:conccon}
We extend $\conc\cdot$ to symbolic contexts with domains
\begin{align*}
  \domnp\conc{[]}&\defeq\Real^m\times\traces_n\\
  \domnp\conc{\scon\,\sterm}&\defeq\domnp\conc{(\lambda y\ldotp\sterm)\,\scon}\defeq\domnp\conc\scon\cap\domnp\conc\sterm\hspace{3em} \\
  \omit\rlap{$\domnp\conc{\PCF f(\sval_1,\ldots,\sval_{\ell-1},\scon,\sterm_{\ell+1},\ldots,\sterm_n)}$}\\
  \omit\rlap{$\defeq\domnp\conc{\sval_1}\cap\cdots\cap\domnp\conc{\sval_{\ell-1}}\cap\domnp\conc\scon\cap\domnp\conc{\sterm_{\ell+1}}\cap\cdots\cap\domnp\conc{\sterm_n}$}\\
  \domnp\conc{\Y\scon}&\defeq\domnp\conc{\Score\scon}\defeq\domnp\conc\scon\\
  \domnp\conc{\Ifleq\scon\sterm\stermb}&\defeq\domnp\conc\scon\cap\domnp\conc\sterm\cap\domnp\conc\stermb
\end{align*}
by
\begin{align*}
  \conc{[]}(\tupf r,\tupf s)&\defeq[]\\
  \conc{\scon\,\sterm}(\tupf r,\tupf s)&\defeq(\conc\scon(\tupf r,\tupf s))\,(\conc\sterm(\tupf r,\tupf s))\\
  \conc{(\lambda y\ldotp\sterm)\,\scon}(\tupf r,\tupf s)&\defeq(\lambda y\ldotp\conc\sterm(\tupf r,\tupf s))\,(\conc\scon(\tupf r,\tupf s))\\
  \omit\rlap{$\conc{\PCF f(\sval_1,\ldots,\sval_{\ell-1},\scon,\sterm_{\ell+1},\ldots,\sterm_n)}(\tupf r,\tupf s)$}\\
  \omit\rlap{$\defeq\PCF f(\conc{\sval_1}(\tupf r,\tupf s),\ldots, \conc{\sval_{\ell-1}}(\tupf r,\tupf s),\conc{\scon}(\tupf r,\tupf s), \conc{\sterm_{\ell+1}}(\tupf r,\tupf s),\ldots, \conc{\sterm_n}(\tupf r,\tupf s))$}\\
  \conc{\Y\scon}(\tupf r,\tupf s)&\defeq\Y(\conc\scon(\tupf r,\tupf s))\\
  \conc{\Ifleq\scon\sterm\stermb}(\tupf r,\tupf s)&\defeq\Ifleq {(\conc\scon(\tupf r,\tupf s))}{(\conc\sterm(\tupf r,\tupf s))}{(\conc\stermb(\tupf r,\tupf s))}\\
  \conc{\Score\scon}(\tupf r,\tupf s)&\defeq\Score(\conc\scon(\tupf r,\tupf s))
\end{align*}
\end{definition}

\correspondenceprop*
\begin{proof}
  Suppose that $\sconfig\sterm\sweight U$ is a symbolic configuration and $(\tupf r,\tupf s)\in U$.

  If $\sterm$ is a symbolic value then $\conc\sterm\args$ is a value \cref{lem:instval} and there is nothing to prove.

  Otherwise, by \cref{lem:basic0},  there exists unique $\scon$ and $\sredex$ such that $\sterm\equiv\scon[\sredex]$. Thus we can define the context $E\equiv\conc\scon\args$ and redex $\redexa\equiv\conc\sredex\args$ (see\ \cref{lem:instred}), and it holds by \cref{eq:substcon}, $\conc\sterm\args\equiv E[\redexa]$.
  \begin{enumerate}
  \item If
  $\sconfig\sredex\sweight U\sred\sconfig\scontra{\sweight'}{U'}$ and $\argsc\in U'$ then by case inspection (see \cref{lem:symcon1aux}), $\config\redexa{\sweight\args}\trace\red\config\contra{\sweight'\argsc}{\trace\concat\traceb}$ such that $\contra\equiv\conc\scontra\argsc$.
  So,
  $\config{E[\redexa]}{\sweight\args}\trace\red\config{E[\contra]}{\sweight'\argsc}{\trace\concat\traceb}$
  and by the substitution \cref{eq:substcon}, $E[\contra]\equiv\conc{\scon[\scontra]}\argsc$.
  \item Conversely, if $\config\redexa{\sweight\args}\trace\red\config\contra{w'}{\tupf{s'}}$ then a simple case analysis (see \cref{lem:consym1aux}) shows that for some $\scontra$, $\sweight'$ and $U'$, $\sconfig\sredex\sweight U\sred\sconfig\scontra{\sweight'}{U'}$ such that $\conc\scontra(\tupf r,\tupf{s'})\equiv\contra$, $\sweight(\tupf r,\tupf{s'})=w'$ and $(\tupf r,\tupf{s'})\in U'$.
    Thus also $\sconfig{\scon[\sredex]}\sweight U\sred^*\sconfig{\scon[\scontra]}{\sweight'}{U'}$ and by the Substitution \cref{eq:substcon}, we have $\conc{\scon[\scontra]}(\tupf r,\tupf s)\equiv E[\redexa]$.
  \end{enumerate}
\end{proof}

\subsection{Supplementary Materials for \Cref{subsec:main result}}
\label{app:diffast}

\lemmavalid*
\begin{proof}
  For the $\mathsf{Score}$-rule this is due to \cref{lem:seva} and the fact that differentiable functions are closed under multiplication. For the other rules differentiability of $\sweight'$ is obvious.

  Furthermore, note that $\mu(\boundary\{(\tupf r,\tupf s\concat[s'])\mid (\tupf r,\tupf s)\in U\land s'\in(0,1)\})=\mu(\boundary U)$ and
  for symbolic values $\sval$,
  \begin{align*}
    \mu\left(\boundary\left(\seva\sval^{-1}(-\infty,0]\right)\right)=\mu\left(\boundary\left(\seva\sval^{-1}(0,\infty)\right)\right)=\mu\left(\boundary\left(\seva\sval^{-1}[0,\infty)\right)\right)=0
  \end{align*}
  because of \cref{lem:seva,ass:pop}.
  \dw{Is this clear enough?}
  Consequently, due to the general fact that ${\boundary(U\cap V)} \subseteq {{\boundary U} \cup {\boundary V}}$ \dw{reference}, in any case, $\mu(\boundary{U'})=0$.
  \end{proof}

\maxtermlemma*
\begin{proof}
  Let $T\in\Borel^m$ be such that $\mu(\Real^m\setminus T)=0$ and for every $\tupf r\in T$, $M[\tupf{\PCF r}/\tupf x]$ terminates almost surely.
  For $\tupf r\in\Real^m$ we use the abbreviations
\[
   \trmaxr\defeq\{\trace\in\traces\mid\args\in\trmax\}\qquad\trtermr\defeq\{\trace\in\traces\mid\args\in\trterm\}
\]
  and we can argue analogously to $\trmax$ and $\trterm$ that they are measurable.
  Similarly to \cref{lem:termleq1}, for all $\tupf r\in\Real^m$, $\mu(\trmaxr)\leq 1$ because $(\trace\concat\traceb)\in\trmaxr$ and $\traceb\neq\emptytrace$ implies $\trace\notin\trmaxr$.

  Therefore, for every $\tupf r\in T$, $\mu(\trmaxr\setminus\trtermr)=0$.
  Finally, due to a consequence of Fubini's theorem (\cref{lem:caval}) and the fact that the Lebesgue measure is $\sigma$-finite,
  \[
    \mu(\trmax\setminus\trterm)=\mu(\{\args\in\Real^m\times\traces\mid \tupf s\in\trmaxr\setminus\trtermr)=0
   \]
\end{proof}

\begin{lemma}
  \label{lem:caval}
  Let $(X,\Sigma_X,\mu)$ and $(Y,\Sigma_Y,\nu)$ be $\sigma$-finite measure spaces. Suppose that $U\in\Sigma_X$ and that for every $r\in X$, $V_r\in\Sigma_Y$, and $W\defeq\{(r,s)\in X\times Y\mid s\in V_r\}$ is measurable.

  If  $\mu(X\setminus U)=0$ and for every $r\in U$, $\mu(V_r)=0$ then $\mu(W)=0$.
\end{lemma}
\begin{proof}
  Let $X_n\in\Sigma_X$ and $Y_n\in\Sigma_Y$ (for $n\in\Nat$) be such that $X=\bigcup_{n\in\Nat}X_n=X$, $Y=\bigcup_{n\in\Nat} Y_n$ and $\mu(X_n)=\nu(Y_n)<\infty$ for every $n\in\Nat$. Define $W_n\defeq W\cap (X_n\times Y_n)$. Clearly $(\mu\times\nu)(W_n)$ is finite.

  By assumption the characteristic function $\charfn W\from X\times Y\to\pReal$ is measurable. By Fubini's theorem \cite[Thm.~1.27]{K02}, for every $n\in\Nat$,
  \begin{align*}
    \mu(W_n)=\int_{{X_n}\times {Y_n}}(\diff(\mu\times\nu))\charfn W=\int_{X_n}(\diff\mu)\int_{Y_n}(\diff\nu)\charfn W
    =\int_{U\cap X_n}(\diff\mu)\bblambda r\ldotp\nu(V_r)
        =0
  \end{align*}
  The third equation is due to $\mu({X_n}\setminus U)=0$.
  The claim is immediate by $W=\bigcup_{n\in\Nat}W_n$.
\end{proof}

\fi



\end{document}